\documentclass[12 pt, onecolumn]{IEEEtran}
\usepackage{amsmath,amssymb,euscript,yfonts,psfrag,latexsym,dsfont,graphicx}
\usepackage{bbm,color,amstext,wasysym,flushend,parskip,balance}
\usepackage{amsthm}

\usepackage{caption}
\usepackage{subcaption}
\graphicspath{{./},{./figures/}}
\usepackage{soul}

\newtheorem{theorem}{Theorem}

\newtheorem{problem}[theorem]{Problem}

\newcommand{\R}{{\mathbb R}}
\newcommand{\cC}{{\mathcal C}}

\newcommand{\cE}{{\mathcal E}}

\newcommand{\cK}{{\mathcal K}}
\newcommand{\cL}{{\mathcal L}}
\renewcommand{\L}{{\mathcal L}}

\newcommand{\cP}{{\mathcal P}}

\newcommand{\cX}{{\mathcal X}}

\newcommand{\bx}{{\mathbf x}}
\newcommand{\by}{{\mathbf y}}

\newcommand{\bX}{{\mathbf X}} %{{\bi{X}}}
\newcommand{\bP}{{\mathbf P}} %{{\bi{P}}}
\newcommand{\bR}{{\mathbf R}} %{{\bi{R}}}
\newcommand{\bpi}{{\boldsymbol\pi}}

\DeclareOldFontCommand{\rm}{\normalfont\rmfamily}{\mathrm}

\def\spacingset#1{\def\baselinestretch{#1}\small\normalsize}
\spacingset{1}

\definecolor{grey}{rgb}{0.6,0.6,0.6}
\definecolor{lightgray}{rgb}{0.97,.99,0.99}
\definecolor{orange}{rgb}{1,.49,0}
\definecolor{darkblue}{rgb}{0.,0.,0.6}

\newcommand{\black}{\color{black}}

\begin{document}
\title{
The most likely evolution of diffusing and vanishing particles:  Schr\"odinger Bridges with unbalanced marginals
}

\author{Yongxin Chen, Tryphon T. Georgiou and Michele Pavon
\thanks{Y.\ Chen is with the School of Aerospace Engineering, Georgia Institute of Technology, Atlanta, GA 30332, USA} 
\thanks{T.T.\ Georgiou is with the Department of Mechanical and Aerospace Engineering, University of California, Irvine, CA 92697, USA}
\thanks{M. Pavon is with the Department of Mathematics ``Tullio Levi-Civita", Universit\`a di Padova, 35121 Padova, Italy} }

\maketitle

\begin{abstract}
Stochastic flows of an advective-diffusive nature are ubiquitous in biology and the physical sciences. Of particular interest is the problem to reconcile observed marginal distributions with a given prior posed by E.\ Schr\"odinger in 1932/32 and known as the Schr\"odinger Bridge Problem (SBP). It turns out that Schr\"odinger's problem can be viewed both as a modeling as well as a control problem. Due to the fundamental significance of this problem, interest in SBP and in its deterministic (zero-noise limit) counterpart of Optimal Mass Transport (OMT) has in recent years enticed a broad spectrum of disciplines, including physics, stochastic control, computer science, probability theory, and geometry.
Yet, while the mathematics and applications of SBP/OMT have been developing at a considerable pace, accounting for marginals of unequal mass has received scant attention; the problem to interpolate between ``unbalanced'' marginals has been approached by introducing source/sink terms into the transport equations, in an adhoc manner, chiefly driven by applications in image registration.

Nevertheless, losses are inherent in many physical processes and, thereby, models that account for lossy transport may also need to be reconciled with observed marginals following Schr\"odinger's dictum; that is, to adjust the probabilty of trajectories of particles, including those that do not make it to the terminal observation point,  so that
the updated law represents the  {\em most likely way that particles may have been transported, or  vanished,   at some intermediate point}. Thus, the purpose of this work is to develop such a natural generalization of the SBP for stochastic evolution with losses, whereupon particles are ``killed''  (jump into a coffin/extinction state) according to a probabilistic law, and thereby mass is gradually lost along their stochastically driven flow. Through a suitable embedding we turn the problem into an {\em SBP for stochastic processes that combine diffusive and jump characteristics}.  Then, following a large-deviations formalism in the style of E.\ Schr\"odinger, given a prior law that allows for losses, we ask for the most probable evolution of particles along with the most likely killing rate as the particles transition between the specified marginals. Our approach differs sharply from previous work involving a Feynman-Kac multiplicative 
reweighing of the reference measure: The latter, as we argue, is far from Schr\"odinger's quest. An iterative scheme, generalizing the celebrated Fortet-IPF-Sinkhorn algorithm, permits to compute the new drift {\em and} the new killing rate of the path-space solution measure.  We finally formulate and solve a related fluid-dynamic control problem for the flow of one-time marginals were both the drift and the new killing rate play the role of control variables. 
\end{abstract}

\section{Introduction}

Consider the problem of estimating the velocity field of oceanic currents by releasing into the water 
a cloud of tracer particles and by sampling their distribution at a later time. The diffusion coefficient is assumed known and the original cloud that is released at time $t=0$ consists of $N$ particles. These are expected to remain in suspension for a duration of time while they diffuse and drift with the current. 
At time $t=1$, their distribution is sampled again. Some of the particles in the meantime have sunk, so that the number of found particles is less than  $N$.
Suppose this experiment is performed several times, treating the model originating from previous experiments as a ``prior". Is it conceivable to ``improve" a prior model in a rational way? More explicitly, by relying on a prior model and the new sampling result, is it possible to determine an updated model that represents the most probable way that the tracer cloud may have been transported?
\black

At first sight, this problem appears to be of a different nature than those treated in the theory of Large Deviations \cite{varadhan1966asymptotic,varadhan1984large,DemZei09}, in that the sought path-space measure is not a probability measure per se. Nevertheless,  in spite of the paucity of the available data, it is possible to solve this inverse problem by a natural embedding technique. A  byproduct is a physically motivated framework to interpolate distributions of unequal mass (integrals). The blueprint for the rationale in our work has been provided by the celebrated duo of papers by E.\ Schr\"odinger in 1931/32 \cite{Sch31,Sch32} where he considered the problem of reconciling marginal distributions with a prior stochastic evolution. 

The original Schr\"odinger Bridge Problem (SBP) asks for the most likely evolution of stochastic  particles as they travel between marginal probability densities $\rho_0$ and $\rho_1$, specified at two points in time (taken as $t_0=0$ and $t_1=1$ without loss of generality), when these marginals fail to be consistent with a known {\em prior law}. 
Interestingly, Schr\"odinger considered this abstract problem before a theory of continuous parameter stochastic processes was in place, and had only been preceded by Ludwig Boltzmann \cite{boltzmann1877uber}. Schr\"odinger attacks the problem very much in his countryman's style, through coarse graining, and applying the De Moivre-Stirling formula and Lagrange multipliers. In spite of the lack of proper tools (Sanov's theorem \cite{sanov1961probability}
will be published in Russian only in 1957), he arrives at the correct answer \cite{Sch31,Sch32}, that the most likely evolution is obtained by solving a certain two-point boundary value problem (Schr\"odinger system of equations).  Important contributions to this theory are then provided by Fortet, Beurling and Jamison \cite{fortet1940resolution,beurling1960automorphism,Jam74,jamison1975markov}. It took more than half a century  before F\"ollmer \cite{Fol88}, recovering Schr\"odinger's original motivation, properly cast the problem within the framework of Large-Deviations theory. The field has since seen several other significant contributions, a partial selection being \cite{zambrini1986stochastic,nagasawa1989transformations,wakolbinger1989simplified,Bla92,dawson1990schrodinger,nagasawa1990stochastic,Wak90,aebi1992large,Mik04,MikThi08,CheGeoPav14a,Leo12,Leo14,CheGeoPav14e,SIREV,Con19}. Here \cite{Wak90,Leo14,SIREV} are survey papers. Observe that, in view of  Sanov's theorem \cite{DemZei09},   the SBP amounts to seeking a new probability law on the path space of the stochastic particles that is consistent with the given marginals and, at the same time, is the closest to the prior probability law in the relative entropy sense.

Schr\"odinger's Bridge Problem (SBP), as well as its zero-noise limit of Optimal Mass Transport (OMT), continue to impact a growing range of disciplines and applications. In this expanding mathematical landscape, the problem to account for variable mass along the transport path received attention from early on. It was  chiefly motivated by the need to interpolate distributions of unequal mass for times series spectral analysis and image registration  \cite{koehl2021physics}.
The viewpoint that is being pursued herein is closer in spirit to the original rationale of E.\ Schr\"odinger as we build on a Large Deviations formalism. To this end, we consider below a diffusion process with killing and seek the closest update of the corresponding law that is in agreement with the marginal data.
Thus, we ask for the most likely evolution of stochastic particles which are known to obey a given {\em prior law with potential for losses} (``killing rate'') while they transition between two marginal distributions $\rho_0$ and $\rho_1$ as before. The two distributions are not necessarily consistent with the prior law and neither is the loss of mass necessarily consistent with the prior killing rate. 

In our formulation of the unbalanced Schr\"odinger Bridge Problem (uSBP), the marginals cannot be assumed to be probability distributions as their integrals differ due to losses. To this end, we embed the distributions into a frame that includes a coffin/extinction state, leading to a probability law on a continuum together with a discrete state.  Thereupon, we find the updated law and killing rate that minimize the relative entropy to the prior with losses, and are consistent with the two marginals. In the special case when the marginals are already consistent with the prior, naturally, the solution coincides with the lossy prior, differently from what happens in other formulations of SBP with killing which are based on Feyman-Kac functionals \cite{nagasawa1990stochastic,wakolbinger1989simplified,Bla92,dawson1990schrodinger,aebi1992large,leonard2011stochastic,CheGeoPav14c,CheGeoPav17a}, and unbalanced transport \cite{chizat2018scaling,chizat2018unbalanced,chen2019interpolation,koehl2021physics}, as discussed in Section \ref{sec:killing}.

The structure of the paper is as follows. In Section \ref{sec:SBP} we revisit classical Schr\"odinger bridge problems. The main framework on Schr\"odinger bridges with unbalanced marginals is presented in Section \ref{sec:uSBP}. We also present a fluid dynamic formulation of the main framework in Section \ref{sec:fluid}. A comprehensive comparison between Schr\"odinger bridges with unbalanced marginals and existing results on Schr\"odinger bridges with killing is provided in Section \ref{sec:killing}. This is followed by a numerical example in Section \ref{sec:example} and a concluding remark in Section \ref{sec:conclusion}. 

\section{Preliminaries on Schr\"odinger Bridge Problem}\label{sec:SBP}
We briefly review elements of the theory of Schr\"odinger Bridge Problem (SBP). 
To this end, consider a diffusion process
	\begin{equation}\label{eq:diffusion}
		dX_t = b(t,X_t)dt + \sigma(t,X_t) dW_t,
	\end{equation}
over the Euclidean space $\R^n$. 
In Schr\"odinger's original thought experiment, a large number $N$ of trajectories over the time interval $[0,1]$ are independently sampled from \eqref{eq:diffusion} with probability distribution of $X_t$ at the initial time $t=0$ being $\rho_0$. The law of large numbers dictates that the terminal distribution at time $t=1$ must be
(approximately)
	\begin{equation}\label{eq:largenumber}
		\int_{\R^n} q(0,x,1,\cdot) \rho_0(x) dx,
	\end{equation}
where $q(t,x,s,y),\, t<s$ denotes the kernel of transition rates from state $x$ at time $t$ to state $y$ at time $s$. 
Now suppose the observed marginal distribution  at time $t=1$, denoted by $\rho_1$, is inconsistent with \eqref{eq:largenumber} and the prior kernel $q(0,x,1,y)$, that is,
	\[
		\rho_1(\cdot) \neq \int_{\R^n} q(0,x,1,\cdot) \rho_0(x) dx.
	\]
Schr\"odinger's problem then seeks the most likely evolution that the particles may have taken between the specified marginals. That is, SBP seeks a suitable update of the law of the diffusion process that reconciles the two marginals $\rho_0,\rho_1$. In the sequel and for notational simplicity, we use the same symbol $\rho$ to denote both, the probability density, as well as the corresponding measure $d\rho= \rho dx$, depending on the context. 

As first noted by F\"ollmer \cite{Fol88}, SBP can be more clearly expressed in the language of the theory of {\em large deviations} \cite{DemZei09}. Specifically, let $\Omega = C([0,1],\R^n)$ denote the space of continuous functions on $[0,1]$ with values in $\R^n$, and $\cP(\Omega)$ denote the space of probability laws over $\Omega$.  Given any two probability measures $P,Q$, the {\em relative entropy} of $P$ with respect to $Q$ is 
	\begin{equation}
		H(P\mid Q) = 
		\begin{cases}
			\int dP \log \frac{dP}{dQ} & \mbox{if}~~ P \ll Q
			\\
			+\infty & \mbox{otherwise}.
		\end{cases}
	\end{equation}
Now consider $N$ independent trajectories 
 $X_t^1, X_t^2,\ldots, X_t^N \in \Omega$ 
of a diffusion having law $R\in\cP(\Omega)$, and let $L_N$ denote their empirical distribution.
Then, asymptotically as $N\to\infty$, Sanov's theorem\footnote{ Sanov's theorem holds when the process takes values in any Polish space.} gives the exponential rate of decay for the probability of occurence of an empirical distribution that differs from the law $R$ \cite{DemZei09} as
	\begin{equation}\label{eq:sanov} 
		{\rm Prob} (L_N \in A) \approx \exp (-N \inf_{P\in A} H(P\mid R)), ~~\forall A \subset \cP(\Omega).
	\end{equation}
Thus, Sanov's result expresses the likelihood of observing an empirical distribution approximated by $P$ in terms of the relative entropy $H(P\mid R)$. Thence, SBP can be formulated as follows:\\

\begin{problem}
Let $R\in\cP(\Omega)$ be the probability measure on $\Omega$ induced by the prior process \eqref{eq:diffusion} with initial distribution $\rho_0$. Determine
	\begin{equation}\label{eq:SBP}
		P^\star := {\rm arg}\min_{P\in \cP(\Omega)} \left\{H(P\mid R)~\mid~ P_0 = \rho_0, P_1 = \rho_1\right\},
	\end{equation}
	where $P_t$ denotes the marginal distribution of $P$ at time $t$ (i.e., the push forward $X_t\#P=P_t$).
\end{problem}

The entropy functional is strictly convex which ensures uniqueness of the minimizer (when it exists). Further, if $R_0=\rho_0$ as well as $R_1 =\rho_1$, the solution to SBP coincides (trivially) with the prior law $R$, i.e., $P^\star=R$, achieving the minimal value $H(P^\star\mid R)=0$. When
$R_1\neq \rho_1$, the SBP thus seeks an updated law $P^\star$ that is closest to the prior in the sense of relative entropy and restores consistency with the marginals (which fails for the prior $R$).
Next we briefly discuss the solution to SBP. For an in depth exposition see \cite{Leo14} and the review articles \cite{chen2021optimal,SIREV}.

By disintegration of measure,
\begin{align*}
		R(\cdot) &= \int_{\R^n\times \R^n} R^{xy}(\cdot) R_{01}(dxdy), \mbox{ and}\\
		P(\cdot) &= \int_{\R^n\times \R^n} P^{xy}(\cdot) P_{01}(dxdy),
\end{align*}	
where $R_{01}$ ($P_{01}$) denotes the joint marginal distribution of $R$ ($P$) of $X_t$ for $t\in\{0,1\}$ (i.e., $R_{01}=X_{01}\# R$, and similarly, for $P$),
and $R^{xy}$ ($P^{xy}$) denotes the measure induced by $P$ conditioned on $(X_0 = x, X_1 = y)$.
It follows that 	
	\begin{equation}\label{eq:Hdisintegration1}
		H(P\mid R ) =  H(P_{01} \mid R_{01}) + \int H(P^{xy} \mid R^{xy}) P_{01}(dxdy).
	\end{equation}
Clearly,  when $P^{xy} = R^{xy}$ for any $x,y\in \R^n$, the second term on the right assumes the minimal value $0$. An immediate consequence is the following {\em static} formulation of the SBP.\\

\begin{problem} Determine
	\begin{equation}\label{eq:SSBP}
		\pi^\star := {\rm arg}\min_{\pi\in \cP(\R^n\times \R^n)} \left\{H(\pi\mid R_{01})~\mid~ \pi_0 = \rho_0, \pi_1 = \rho_1\right\}.
	\end{equation}
\end{problem}
To distinguish between the two formulations \eqref{eq:SBP} and \eqref{eq:SSBP}, we refer to \eqref{eq:SBP} as the {\em dynamic} SBP. The two formulations are equivalent in the sense that solving one provides a solution to the other, as noted next.\\

\begin{theorem}[\cite{Leo14}]\label{thm:staticdynamic1}
Suppose $P^\star$ is a solution to the dynamic SBP \eqref{eq:SBP}, then $P^\star_{01}$ solves the static SBP \eqref{eq:SSBP}. On the other hand, if $\pi^\star$ is a solution to \eqref{eq:SSBP}, then setting $P^\star=\int_{\R^n\times \R^n} R^{xy}(\cdot) \pi^\star(dxdy)$ solves \eqref{eq:SBP}, while $P^\star_{01}=\pi^\star$.
\end{theorem}
\begin{proof} It follows readily from \eqref{eq:Hdisintegration1}.
\end{proof}

A direct consequence of Theorem \ref{thm:staticdynamic1} is that the Radon-Nikodym ratios between solutions and priors for the two problems, the static \eqref{eq:SUSBP} and the dynamic \eqref{eq:USBP}, coincide, namely,
	\begin{equation}\label{eq:nikodym}
		\frac{dP^\star}{dR} = \frac{d\pi^\star}{dR_{01}} (X_0,X_1).
	\end{equation}
In fact, this ratio can be factored into two parts, one that depends only on $X_0$ and one that depends on $X_1$, as follows.\\

\begin{theorem}[\cite{Leo14}]\label{prop:solution1}
Assume that $R_{01} \ll R_0\otimes R_1$, and that there exists $\pi\in \cP(\R^n\times \R^n)$ such that $\pi_0 = \rho_0$, $\pi_1 = \rho_1$ (i.e., feasible), for which $H(\pi \mid R_{01})<+\infty$. Then the static problem \eqref{eq:SSBP} admits a unique solution $\pi^\star$ and there exist two measurable functions $f, g: \cX \rightarrow \R_+$ such that
	\begin{equation}\label{eq:optpi}
		\pi^\star = f(X_0) g(X_1) R_{01}.
	\end{equation}
The two factors $f, g$ are solutions to the Schr\"odinger system
	\begin{subequations}\label{eq:SBsysC1}
	\begin{eqnarray}
		\frac{d\rho_0}{dR_0}(x) &=& f(x) R(g(X_1)\mid X_0 = x),
		\\
		\frac{d\rho_1}{dR_1}(y) &=& g(y) R(f(X_0) \mid X_1 = y).
	\end{eqnarray}
	\end{subequations}
Moreover, the unique solution to the dynamic problem \eqref{eq:SBP} is 
	\begin{equation}\label{eq:optP}
		P^\star =  f(X_0) g(X_1) R.
	\end{equation}
\end{theorem}

The above theorem provides an abstract construction of the sought probability law(s) via the solution of the Schr\"odinger system \eqref{eq:SBsysC1}. The local characteristics and the modified stochastic differential equation for the process with law $P^\star$ follow. Computationally, these can be expressed most succinctly in terms of a pair of two, forward and backward in time (and identical to that of the prior Fokker-Planck and its adjoint) equations, that are nonlinearly coupled through boundary conditions. We explain this next.

First recall that the  marginals $R_t(x)$ of the prior law $R$ for the diffusion \eqref{eq:diffusion} satisfy (weakly)  the Fokker-Planck equation
\begin{equation}\label{eq:FK1}
		\partial_t  R_t + \nabla\cdot(b R_t) = \frac{1}{2} \sum_{i,j=1}^n \frac{\partial^2 (a_{ij} R_t)}{\partial x_i\partial x_j}.
	\end{equation}
In what follows, $a(t,x) = \sigma(t,x)\sigma(t,x)'$ is assumed to be everywhere positive definite. Let the two end-point marginals be absolutely continuous with densities $\rho_0$ and $\rho_1$, respectively. The Schr\"odinger system \eqref{eq:SBsysC1} can be reparametrized in terms of
	\begin{subequations}
	\begin{align}
		\hat\varphi(0,x)&:=f(x)R_0(x)\\
		\varphi(1,y)&:=g(y),
		\end{align}
	\end{subequations}
and takes the form
	\begin{subequations}\label{eq:SBsystem1}
	\begin{eqnarray}\label{eq:SBsystema1}
		\partial_t \hat\varphi &=& - \nabla\cdot(b \hat\varphi) +\frac{1}{2} \sum_{i,j=1}^n \frac{\partial^2 (a_{ij} \hat \varphi)}{\partial x_i\partial x_j} 
		\\\label{eq:SBsystemc1}
		\partial_t\varphi &=& - b \cdot\nabla\varphi - \frac{1}{2} \sum_{i,j=1}^n a_{ij}\frac{\partial^2 \varphi}{\partial x_i\partial x_j}
		\\\label{eq:SBsysteme1}
		\rho_0 &=& \varphi(0,\cdot)\hat\varphi(0,\cdot)
		\\\label{eq:SBsystemf1}
		\rho_1 &=& \varphi(1,\cdot)\hat\varphi(1,\cdot).
\end{eqnarray}
	\end{subequations}
	
\begin{theorem}\label{thm:SBEsys1} Let $R$ be the law of \eqref{eq:diffusion} with
$a(t,x) = \sigma(t,x)\sigma(t,x)'$ being positive definite for all $(t,x)\in\R\times \R^n$, and assume that $\rho_0,\rho_1$ are absolutely continuous with respect to the Lebesgue measure. 
There exists a unique (up to a constant positive scaling) pair $(\hat\varphi(t,x), \varphi(t,x))$ of non-negative functions that satisfies the Schr\"odinger system \eqref{eq:SBsystem1}. Moreover, the law $P^\star$ for the dynamic problem \eqref{eq:SBP} is law of the diffusion
	\begin{equation}\label{eq:controldiffusion}
		dX_t = (b(t,X_t) +a(t,X_t)\nabla\log\varphi(t,X_t))dt + \sigma(t,X_t) dW_t,
	\end{equation}
with distribution of $X_0$ being $\rho_0$, and at any time $t\in[0,1]$, the marginal density for $P^\star$ satisfies the identity $p_t(x)=\varphi(t,x)\hat\varphi(t,x)$.\black
\end{theorem}

Existence and uniqueness of solutions for the Schr\"odinger system have been provided in various degrees of generality by Fortet, Beurling and Jamison \cite{fortet1940resolution,beurling1960automorphism,Jam74,jamison1975markov}. 
For a detailed exposition of the theory of Schr\"odinger's problem we refer to Leonard \cite{Leo14}, in particular, \cite[Theorems 2.8, 2.9, 3.4]{Leo14}. \black
A more recent account along with a proof that is based on the contractiveness of suitable maps in the Hilbert metric was given in \cite{chen2016entropic}. In the present paper, we follow a similar approach as in \cite{chen2016entropic}
when analyzing the more general Schr\"odinger system for diffusions with losses and, therefore, we sketch key steps for this more general case that we consider. An added benefit in recasting the Schr\"odinger system as in 
\eqref{eq:SBsystem1} is that it leads, after discretization, to an efficient algorithm for computing $(\hat\varphi(t,x), \varphi(t,x))$, and thereby, $f,g$ as well as $P^\star$. The discretized version of the Schr\"odinger system \eqref{eq:SBsystem1} amounts to the celebrated Sinkhorn algorithm for matrix scaling \cite{SIREV}. 

\section{Unbalanced stochastic transport}\label{sec:uSBP}
We now analyze stochastic flows between unequal marginals following E.\ Schr\"odinger original rationale that is rooted in large deviations theory. To this end, we consider a diffusion process with killing and seek the closest update of the corresponding prior law that restores agreement with marginal data.

Once again consider the diffusion process \eqref{eq:diffusion} but, this time, with a nonnegative killing rate $V(t,x)$ (assume $V(\cdot,\cdot)$ is continuous and not constantly zero). 
A thought experiment similar to Schr\"odinger's, calls for a large number $N$ of trajectories over a time interval $[0,\,1]$, that are independently sampled from \eqref{eq:diffusion} with initial probability distribution $\rho_0$, and a recorded empirical distribution for the surviving particles at time $t=1$ approximated by $\rho_1$, which is inconsistent with the prior law, that is,
	\[
		\rho_1(\cdot) \neq \int_{\R^n} q(0,x,1,\cdot) \rho_0(x) dx.
	\]
	The kernel $q(0,x,1,y)$ is no longer a probability kernel in that $\int_{\R^n} q(0,x,1,y) dy\neq 1$, in general, and thus, neither $\int_{\R^n} q(0,x,1,\cdot) \rho_0(x) dx$ nor $\rho_1$ are necessarily probability densities, due to killing. In particular, 
$\int \rho_1(x)dx=N_s/N\le 1$ where $N_s$ denotes the number of survival particles at time $1$. Just as in the standard SBP,
we consider continuous distributions, assuming that $N$ is large, and seek to identify the most likely behavior of the particles. By behavior we mean the most likely evolution of the particles along with the most likely times that the particles may have gotten killed (or, absorbed by an underlying medium). 

As in the standard Schr\"odinger bridge, the problem arising from the above thought experiment can be formally stated using the theory of {\em large deviations} \cite{DemZei09}. However, in this case, the space of trajectories needs to be modified to accommodate for possible killing of particles. To this end, we augment the state space  of the diffusion $\R^n$ with a ``coffin state'' $\mathfrak c$, resulting in the state space 
\[
\cX=\R^n\cup \{\mathfrak c\}.
\]
Let $\mathbf{\Omega}=D([0,1],\cX)$ be the Skorokhod space over $\cX$, that is, each element in $\mathbf{\Omega}$ is a c\`adl\`ag over $\cX$ \cite{Bil99}. Denote by $\cP(\mathbf \Omega)$ and $\cP(\cX)$ the spaces of probability distributions over $\mathbf \Omega$ and $\cX$, respectively. Each diffusion process $X_t$ ($t\in[0,1]$) on $\R^n$ with killing corresponds to a process $\bX_t$ taking values in $\mathcal X$, and therby, to a law in $\cP(\mathbf \Omega)$. 

Evidently,
$\cX$ is a Polish space. The space $\mathbf{\Omega}$ of c\`adl\`ag over $\cX$ is thus, with the appropriate topology, also a Polish space \cite{Pol12,EthKur09}. Sanov's theorem applies to measures on Polish spaces and, therefore, the likelihood function is once again expressed in terms of the relative entropy between probability laws.
In our unbalanced SBP setting, the set of probability laws over path space $\cP(\mathbf \Omega)$ that are in alignment with the observations is 
	\[
		\{\bP\in \cP(\mathbf \Omega)\mid \bP_0 = p_0,~\bP_1 = p_1\},
	\]
where $p_0, p_1$ are the natural augmentation of $\rho_0, \rho_1$ so that they belong in $\cP(\cX)$, respectively. Specifically, assuming that $\int_{\mathbb R^n} \rho_1(x)dx=1$, we set 
	\begin{subequations}
	\begin{equation}
		p_0=(\rho_0(\cdot),0)
	\end{equation}
 and 
	\begin{equation} 
		p_1=(\rho_1(\cdot),1-\int_{\mathbb R^n} \rho_1(x)dx).
 	\end{equation}
	\end{subequations}
 
Thus, we arrive at the following recasting of uSBP  as an ordinary SBP.\\

\begin{problem}[Unbalanced Schr\"odinger Bridge Problem (uSBP)] Determine
	\begin{equation}\label{eq:USBP}
		\bP^\star := {\rm arg}\min_{\bP\in \cP(\mathbf \Omega)} \left\{H(\bP\mid \bR)~\mid~ \bP_0 = p_0, \bP_1 = p_1\right\}.
	\end{equation}
\end{problem}

As before, verbatim, $\bR(\cdot) = \int_{\cX^2} \bR^{xy}(\cdot) \bR_{01}(dxdy)$ and
$\bP(\cdot) = \int_{\cX^2} \bP^{xy}(\cdot) \bP_{01}(dxdy)$, where
now $\bR_{01}$ ($\bP_{01}$) denotes the joint marginal distribution of $\bR$ ($\bP$) over the marginal $\bX_{0,1}$,
and $\bR^{xy}$ ($\bP^{xy}$) denotes the law conditioned on $\bX_0 = x\in \mathcal X$ and  $\bX_1 = y\in\mathcal X$.
As before, the relation to the static SBP emerges.\\

\begin{problem} Determine
	\begin{equation}\label{eq:SUSBP}
		\bpi^\star:={\rm arg}\min_{\bpi\in \cP(\cX^2)} \left\{H(\bpi\mid \bR_{01})~\mid~ \bpi_0 = p_0, \bpi_1 = p_1\right\}.
	\end{equation}
\end{problem}

The two formulations are once again equivalent, as it readily follows from the identity $H(\bP\mid \bR ) =  H(\bP_{01} \mid \bR_{01}) + \int_{\cX^2} H(\bP^{xy} \mid \bR^{xy}) \bP_{01}(dxdy)$.\\

\begin{theorem}\label{thm:staticdynamic}
Suppose $\bP^\star$ solves the dynamic uSBP \eqref{eq:USBP}, then $\bP^\star_{01}$ also solves the static uSBP \eqref{eq:SUSBP}. On the other hand, if $\bpi^\star$ solves \eqref{eq:SUSBP}, then setting $\bP^\star=\int_{\cX^2} \bR^{xy}(\cdot) \bpi^\star(dxdy)$ solves \eqref{eq:USBP}, while $\bP^\star_{01}=\bpi^\star$.
\end{theorem}

The Radon-Nikodym ratios between solutions and priors for the two problems, analogous to \eqref{eq:nikodym}, applies here too, and the analogous expressions for the Schr\"odinger system in Theorem \ref{prop:solution1} follow as well. More explicitly, the solutions to \eqref{eq:USBP} and \eqref{eq:SUSBP} are of the form
	\begin{equation}
		\bP^\star = f(\bX_0)g(\bX_1)\bR
	\end{equation}
and 
	\begin{equation}
		\bpi^\star = f(\bX_0)g(\bX_1)\bR_{01}
	\end{equation}
respectively.
The divergence between the standard SBP and the present uSBP becomes noticeable when we seek explicit solutions via analogues of system \eqref{eq:SBsystem1} and of the corresponding Fokker-Plank equation in Theorem \ref{thm:SBEsys1}, since now, we need to specify the update on the prior killing rate. We detail this next.

\subsection{Generalized Schr\"odinger system}

The Fokker-Planck equation for a diffusion \eqref{eq:diffusion} with killing rate $V(t,x)$
is
	\begin{equation}\label{eq:FK}
		\partial_t  R_t+ \nabla\cdot(b R_t) + V R_t = \frac{1}{2} \sum_{i,j=1}^n \frac{\partial^2 (a_{ij} R_t)}{\partial x_i\partial x_j}.
	\end{equation}
As before, $a(t,X) = \sigma(t,X)\sigma(t,X)'$ is assumed to be positive definite throughout. The corresponding Schr\"odinger system and its relation to the law of $\bP^\star$ can be expressed after reparametrizing the pair $(f,g)$ of functions on $\mathcal X$ as follows
	\begin{subequations}\label{eq:newparameters}
	\begin{eqnarray}
		f(x)\bR_0(x) &=&
		\begin{cases}
			\hat\varphi(0,x) & \mbox{if}~x\in \R^n
			\\
			\hat\psi(0) & \mbox{if}~x=\mathfrak c,
		\end{cases} 
		\\
		g(y) &=&
		\begin{cases}
			\varphi(1,y) & \mbox{if}~y\in \R^n
			\\
			\psi(1) & \mbox{if}~y=\mathfrak c.
		\end{cases} 
	\end{eqnarray}
	\end{subequations}
	Comparing with \eqref{eq:SBsystem1}, the Schr\"odinger system along with the non-linear coupling constraints now becomes
	\begin{subequations}\label{eq:SBsystem}
	\begin{eqnarray}\label{eq:SBsystema}
		\partial_t \hat\varphi &=& - \nabla\cdot(b \hat\varphi)-V\hat\varphi +\frac{1}{2} \sum_{i,j=1}^n \frac{\partial^2 (a_{ij} \hat \varphi)}{\partial x_i\partial x_j} 
		\\\label{eq:SBsystemb}
		\frac{d\hat\psi}{dt} &=& \int V\hat\varphi(t,x) dx
		\\\label{eq:SBsystemc}
		\partial_t\varphi &=& - b \cdot\nabla\varphi+V\varphi - \frac{1}{2} \sum_{i,j=1}^n a_{ij}\frac{\partial^2 \varphi}{\partial x_i\partial x_j} -V\psi
		\\\label{eq:SBsystemd}
		\frac{d\psi}{dt} &=& 0
		\\\label{eq:SBsysteme}
		\rho_0 &=& \varphi(0,\cdot)\hat\varphi(0,\cdot)
		\\\label{eq:SBsystemf}
		\rho_1 &=& \varphi(1,\cdot)\hat\varphi(1,\cdot)
		\\\label{eq:SBsystemg}
		 \hat\psi(0) &=& 0
		 \\\label{eq:SBsystemh}
		 \psi(1)\hat\psi(1) &=& 1-\int \rho_1.
	\end{eqnarray}
	\end{subequations}

\begin{theorem}\label{thm:SBEsys}
Let $R$ be the law of a diffusion \eqref{eq:diffusion} with nontrivial killing rate $V(t,x)$ and
$a(t,x) = \sigma(t,x)\sigma(t,x)'$ being positive definite for all $(t,x)\in\R\times \R^n$, and assume that $\rho_0,\rho_1$ are absolutely continuous with respect to the Lebesgue measure.
There exists a unique (up to a constant positive scaling)  $4$-tuple $(\hat\varphi(t,x), \hat \psi(t), \varphi(t,x), \psi(t))$ of non-negative functions that satisfies the Schr\"odinger system \eqref{eq:SBsystem}.
\end{theorem}

The proof of the theorem, given in Appendix \ref{sec:GSS}, is based on the contractiveness of the iterative scheme that consists in alternating between evaluation of
$(\hat\varphi(1,\cdot),\hat\psi(1))$ from $(\varphi(0,\cdot),\psi(0))$ using (\ref{eq:SBsysteme}-\ref{eq:SBsystema}-\ref{eq:SBsystemg}-\ref{eq:SBsystemb}), and then evaluating $(\varphi(0,\cdot),\psi(0))$ in a followup cycle from $(\hat\varphi(1,\cdot),\hat\psi(1))$ using the backward in time integration, via the remaining equations. Specifically, we prove that the iteration 
    \begin{align}\label{eq:newmaps}
            {\hat\varphi(1,\cdot)\choose{\hat\psi(1)}}
       			{\mapsto} 
       {\varphi(0,\cdot) \choose \psi(0)}
       			{\mapsto}
      {\hat\varphi(1,\cdot)\choose \hat\psi(1)}_{\rm next}
    \end{align}
is strictly contractive in the Hilbert metric.  

As in the ordinary SBP the discretized Schr\"odinger system \eqref{eq:SBsystem} leads to an  efficient algorithm to compute $(\hat\varphi(t,x), \hat \psi(t), \varphi(t,x), \psi(t))$, and thereby, $\bP^\star$ as well as the corresponding Fokker-Planck equation for the corresponding marginals, that is explained next. 

\subsection{Dynamic formulation}

In general, a multiplicative transformation such as $\bR\to \bP^\star=f(\bX_0)g(\bX_1)\bR$, preserves the Markovian character. Moreover, the generators of the respective semi-groups that relate in this way, herein, $\L_t^\bR,\L_t^{\bP^\star}$, can be evaluated from one another directly by utilizing the multiplicative factors and the so-called carr\'e du champ operator
\[
\Gamma_t(u,v):=\L_t(uv)-u\L_t(v)-v\L_t(u).
\]
Specifically (see \cite[Equation (3.6)]{Leo14}, and also \cite{leonard2011stochastic,revuz2013continuous}),
\begin{equation}\label{eq:3_6}
\L_t^{\bP^\star} u (x)= \L_t^{\bR}u (x) + \Gamma_t^{\bR}(g_t,u)(x)/g_t(x),
\end{equation}
where
	\begin{equation}
	g_t(y) =
		\begin{cases}
			\varphi(t,y) & \mbox{if}~y\in \R^n
			\\
			\psi(t) & \mbox{if}~y=\mathfrak c.
		\end{cases} 
	\end{equation}

In light of Theorem \ref{thm:SBEsys} we now establish an explicit characterization of the dynamic unbalanced Schr\"odinger bridge problem \eqref{eq:USBP}. We denote by $P_t$ the marginal of $\bX_t$ restricted to the first component in $\mathcal X$, and by $q_t$ the probabilty of the coffin state. Thus, we use the vectorial notation
\[
 \bP_t=:(P_t,q_t).
\]
Accordingly, for the marginals $\bR_t=(R_t,s_t)$ of the prior, $R_t$ satisfies the Fokker-Planck equation \eqref{eq:FK} while $s_t=1-\int_{\mathbb R^n}R_t(x)dx$. The solution $\bP^\star$ to \eqref{eq:USBP} is then characterized by the following theorem.

\begin{theorem}
The solution $\bP^\star$ to \eqref{eq:USBP} corresponds to a diffusion process
	\begin{equation}\label{eq:controldiffusion}
		dX_t = (b(t,X_t) +a(t,X_t)\nabla\log\varphi(t,X_t))dt + \sigma(t,X_t) dW_t
	\end{equation}
with killing rate $\psi V/\varphi$, where $\varphi$ is obtained from the solution of the generalized Schr\"odinger system \eqref{eq:SBsystem}. Accordingly,
	\begin{equation}\label{eq:newFP}
		 \partial_t P_t + \nabla \cdot((b+ a\nabla \log \varphi)P_t) = \frac{1}{2} \sum_{i,j=1}^n \frac{\partial^2 (a_{ij} P_t)}{\partial x_i\partial x_j} - \frac{\psi}{\varphi} V P_t.
	\end{equation}
\end{theorem}
\begin{proof} The generator of $\bR$, is of the form
	\begin{equation}
		\L^\bR_t \;:\;
 \left[\begin{matrix} \varphi \\ \psi \end{matrix} \right] \mapsto \left[\begin{matrix} b\cdot \nabla \varphi -V \varphi + \frac{1}{2} \sum_{i,j=1}^n a_{ij}\frac{\partial^2 \varphi}{\partial x_i\partial x_j}+V\psi  \\ 0 \end{matrix} \right].
	\end{equation}
The carr\'e du champ operator becomes
	\begin{equation}
		\Gamma^\bR_t (\left[\begin{matrix} \varphi_1 \\ \psi_1 \end{matrix} \right], \left[\begin{matrix} \varphi_2 \\ \psi_2 \end{matrix} \right])
		=\left[\begin{matrix} V\varphi_1\varphi_2 + a \nabla\varphi_1\cdot \nabla\varphi_2 + V \psi_1\psi_2 - V \psi_1 \varphi_2 -V \varphi_1 \psi_2 \\ 0 \end{matrix} \right].
	\end{equation}
We readily obtain that
	\begin{eqnarray}\label{eq:componentwise}
		\L^{\bP^\star}_t  \left[\begin{matrix} v\\ \eta \end{matrix} \right] &=& \L^\bR_t   \left[\begin{matrix} v\\ \eta \end{matrix} \right] + \Gamma^\bR_t  (\left[\begin{matrix} \varphi \\ \psi \end{matrix} \right], \left[\begin{matrix} v \\ \eta \end{matrix} \right])/\left[\begin{matrix} \varphi \\ \psi \end{matrix} \right]
		\\&=&\nonumber
		\left[\begin{matrix} b\cdot \nabla v -V v + \frac{1}{2} \sum_{i,j=1}^n a_{ij}\frac{\partial^2 v}{\partial x_i\partial x_j}+V\eta + V v +a \nabla \log \varphi \cdot \nabla v + V \frac{\psi}{\varphi} \eta - V \frac{\psi}{\varphi} v- V\eta\\ 0 \end{matrix} \right]
		\\&=&\nonumber
		\left[\begin{matrix} (b+a \nabla \log \varphi) \cdot \nabla v- \frac{\psi}{\varphi} V v+ \frac{1}{2} \sum_{i,j=1}^n a_{ij}\frac{\partial^2 v}{\partial x_i\partial x_j}  + \frac{\psi}{\varphi} V  \eta \\ 0 \end{matrix} \right],
	\end{eqnarray}
	where the division in \eqref{eq:componentwise} is carried out componentwise.
The generator is that of the diffusion process \eqref{eq:controldiffusion} with killing rate $\psi V/\varphi$ over the extended state space $\cX$. The Fokker-Planck equation \eqref{eq:newFP} can be obtained by taking the dual of $\L^{\bP^\star}_t$.
\end{proof}

\begin{theorem}
The marginal distribution $\bP_t^\star$ on the first component of $\mathcal X$ is $P_t=\varphi(t,\cdot) \hat\varphi(t,\cdot)$, and on the second component of $\mathcal X$ is $q_t = \psi(t)\hat\psi(t)$.
\end{theorem}
\begin{proof}
We verify that $P_t$ as above satisfies the Fokker-Planck equation associated with the diffusion \eqref{eq:controldiffusion} with killing rate $\psi V/\varphi$. 
To this end, let
 $P_t(\cdot):= \varphi(t,\cdot) \hat\varphi(t,\cdot)$. Then by \eqref{eq:SBsystema} and \eqref{eq:SBsystemc} we obtain
	\begin{eqnarray*}
		0&=&\partial_t P_t + \nabla \cdot((b+ a\nabla \log \varphi)P_t) - \frac{1}{2} \sum_{i,j=1}^n \frac{\partial^2 (a_{ij} P_t)}{\partial x_i\partial x_j} + \hat\varphi \psi V 
		\\&=& \partial_t P_t + \nabla \cdot((b+ a\nabla \log \varphi)P_t) - \frac{1}{2} \sum_{i,j=1}^n \frac{\partial^2 (a_{ij} P_t)}{\partial x_i\partial x_j} + \frac{\psi}{\varphi} V P_t,
	\end{eqnarray*}
which is exactly the desired Fokker-Planck equation \eqref{eq:newFP}. 
Similarly, by  \eqref{eq:SBsystemb} and \eqref{eq:SBsystemd},
	\[
		\frac{dq_t}{dt} = \psi(t) \int V\hat\varphi(t,x)dx = \int \frac{\psi}{\varphi} V P_t dx, 
	\]
which is consistent with $P_t$ and the new killing rate ${\psi}V/{\varphi}$.
\end{proof}

\section{Fluid dynamic formulation}\label{sec:fluid}
The original Schr\"odinger bridge problem, when there is no killing,  is known to be equivalent to the stochastic control problem of minimizing control energy subject to the marginal two end-point constraints \cite{CheGeoPav14e}, or equivalently, to a fluid dynamic formulation whereby the velocity field $u(t,\cdot)$ effecting the flow minimizes this action integral, namely,
	\begin{subequations}\label{eq:fluid}
	\begin{eqnarray}
		\min_{P_t(\cdot), u(t,\cdot)} && \int_0^1 \int_{\R^n} \frac{1}{2} \|u(t,x)\|^2 P_t dx dt
		\\&& \partial_t  P_t+ \nabla\cdot((b+\sigma u) P_t) - \frac{1}{2} \sum_{i,j=1}^n \frac{\partial^2 (a_{ij} P_t)}{\partial x_i\partial x_j}=0
		\\&& P_0 = \rho_0,\quad P_1=\rho_1.
	\end{eqnarray}
	\end{subequations}
The optimization takes place over the {\em feedback control policy}-{\em flow field} $u(t,x)$ together with the corresponding density flow $P_t(x)$. Below, in this section, we derive an analogous 
formulation for the Schr\"odinger bridge problems with unbalanced marginals.

Along the flow, the killing rate may deviate from the prior $V$ and is to be determined. To quantify the deviation of the posterior killing rate from the prior, we introduce an entropic cost
 inside the action integral, to penalize changes in the ratio $\alpha(t,x)$ between the posterior and the prior killing rate. That is, $\alpha$ is an added optimization variable which is  $\alpha(t,x)\ge 0$, and with the posterior killing rate being $\alpha V$. To penalize differences between the posterior and the prior killing rates
we introduce the factor
	\begin{align}\label{eq:entropiccost}
		&\alpha \log \alpha -\alpha +1
	\end{align}
inside the action integral, which is convex and achieves the minimal value $0$ at $\alpha =1$. This entropy cost has been used in \cite{Leo14,Leo16} to study Schr\"odinger bridge problem over graphs. It is associated with the large deviation principle for continuous-time Markov chain with discrete state. Combining this entropic cost term for the ratio of killing rates with \eqref{eq:fluid} we arrive at
	\begin{subequations}\label{eq:fluidkilling}
	\begin{eqnarray}\label{eq:fluidkillinga}
		\min_{P, u,\alpha} && \int_0^1 \int_{\R^n} \left[\frac{1}{2} \|u(t,x)\|^2 P_t +(\alpha \log \alpha -\alpha +1)VP_t \right]dx dt
		\\&& \partial_t  P_t+ \nabla\cdot((b+\sigma u) P_t) +\alpha VP_t- \frac{1}{2} \sum_{i,j=1}^n \frac{\partial^2 (a_{ij} P_t)}{\partial x_i\partial x_j}=0 \label{eq:fluidkillingb}
		\\&& P_0 = \rho_0,\quad P_1=\rho_1.\label{eq:fluidkillingc}
	\end{eqnarray}
	\end{subequations}
Note that the control strategy has now two components, a drift term $u(t,x)$ and a correcting term $\alpha(t,x)$ for the killing rate.\\
\begin{theorem}
Let $(\hat\varphi(t,x), \hat \psi(t), \varphi(t,x), \psi(t))$ be the solution to the Schr\"odinger system \eqref{eq:SBsystem}, then the solution to \eqref{eq:fluidkilling} is given by the choice
	\begin{subequations}\label{eq:thm12}
	\begin{eqnarray}
		u^\star(t,x) &=& \sigma (t,x)'\nabla \log \varphi(t,x)
		\\
		\alpha^\star(t,x) &=& \frac{\psi(t)}{\varphi(t,x)}\\
		P_t(x)&=&\varphi(t,x)\hat\varphi(t,x).
	\end{eqnarray}
	\end{subequations}
\end{theorem}
\begin{proof} 
We verify that conditions \eqref{eq:thm12} ensure stationarity of the Lagrangian for \eqref{eq:fluidkilling}. 
Introducing the Lagrange multiplier $\lambda(t,x)$,  the Lagrangian for \eqref{eq:fluidkilling} is
	\begin{align*}
		\cL = &\int_0^1 \int\left[\frac{1}{2} \|u\|^2 P_t +(\alpha \log \alpha -\alpha +1)VP_t \right.\\
		&\left.+ \lambda \left(\partial_t  P_t+ \nabla\cdot((b+\sigma u) P_t) +\alpha VP_t- \frac{1}{2} \sum_{i,j=1}^n \frac{\partial^2 (a_{ij} P_t)}{\partial x_i\partial x_j}\right)\right]dx dt.
	\end{align*}
Applying integration by part we obtain
	\begin{align}\nonumber
		\cL =& \int_0^1 \int\left[\frac{1}{2} \|u\|^2 P_t +(\alpha \log \alpha -\alpha +1)VP_t - P_t\partial_t\lambda -\nabla\lambda\cdot (b+\sigma u) P_t +\alpha V\lambda P_t
		\right.\\
		&\left.- \frac{1}{2} \sum_{i,j=1}^n a_{ij} \frac{\partial^2 \lambda}{\partial x_i\partial x_j}P_t\right]dx dt
		+ \int \lambda (1,x) P_1(x) dx - \int \lambda(0,x) P_0(x) dx. \label{eq:Lagrangian}
	\end{align}
Minimizing the above over $u$ yields
	\begin{subequations}\label{eq:optcontrol}
	\begin{equation}
		u^\star(t,x) = \sigma' \nabla \lambda.
	\end{equation}
Similarly, minimization over $\alpha$ yields
	\begin{equation}
		\alpha^\star(t,x) = \exp (-\lambda).
	\end{equation}
	\end{subequations}
Substituting \eqref{eq:optcontrol} into \eqref{eq:Lagrangian} we obtain
	\begin{align*}
		\cL =& \int_0^1\int P_t\left(-\frac{1}{2} a \nabla \lambda \cdot \nabla\lambda- b\cdot \nabla\lambda-\partial_t \lambda- \frac{1}{2} \sum_{i,j=1}^n a_{ij} \frac{\partial^2 \lambda}{\partial x_i\partial x_j} +V(1-\exp(-\lambda))\right) dxdt
		\\&+\int \lambda (1,x) P_1(x) dx - \int \lambda(0,x) P_0(x) dx.
	\end{align*}
The optimality condition
	\begin{equation}\label{eq:lambda}
		\partial_t \lambda+ b\cdot \nabla\lambda+ \frac{1}{2} \sum_{i,j=1}^n a_{ij} \frac{\partial^2 \lambda}{\partial x_i\partial x_j} +\frac{1}{2} a \nabla \lambda \cdot \nabla\lambda-V(1-\exp(-\lambda)) = 0
	\end{equation}
follows. Now, let 
	\begin{equation}\label{eq:lambdaphi}
		\lambda (t,x) = \log \frac{\varphi(t,x)}{\psi(t)},
	\end{equation}
then \eqref{eq:lambda} becomes 
	\begin{equation}
		\partial_t \varphi - \frac{d\psi}{dt} \frac{\varphi}{\psi} + b\cdot \nabla \varphi -  V\varphi + \frac{1}{2} \sum_{i,j=1}^n a_{ij}\frac{\partial^2 \varphi}{\partial x_i\partial x_j} +V\psi =0,
	\end{equation}
and by setting $d\psi/dt=0$, the above reduces to \eqref{eq:SBsystemc}. Finally, plugging \eqref{eq:lambdaphi} into \eqref{eq:optcontrol} yields \eqref{eq:thm12}.
\end{proof}

Substituting the optimal control \eqref{eq:thm12} into \eqref{eq:fluidkillingb} yields the closed loop dynamics under optimal control strategy. Clearly, it is the same as \eqref{eq:newFP} associated with the solution $\bP^\star$ to the uSBP \eqref{eq:USBP}.

\section{SBP over reweighted processes}\label{sec:killing}

Some early attempts to formulate the Schr\"odinger Bridge Problem for diffusions with losses date back to Nagasawa and Wakolbinger \cite{nagasawa1990stochastic,wakolbinger1989simplified}. These focused on processes that are suitably reweighed via a Feynman-Kac multiplicative functional to model losses. Earlier relevant work on Schr\"odinger Bridges over reweighed processes includes \cite{nagasawa1990stochastic,wakolbinger1989simplified,Bla92,dawson1990schrodinger,aebi1992large,leonard2011stochastic,CheGeoPav14c}. In particular, e.g., \cite[Section 8]{wakolbinger1989simplified}, and more recently, \cite{leonard2011stochastic} discuss Feynman-Kac reweighing of the prior measure $R$, into $f(X_0)\exp\left(-\int_{0}^1V(t,X_t)dt\right)g(X_1)R$. Such a process, with this special Radon-Nikodym derivative, is referred to as the $h$-transform of $R$. 
To distinguish this prior work from our uSBP formulation, we refer to the earlier formulation as {\em SBP over reweighted processes}.

Let $\hat\rho_1$ be a normalized version of $\rho_1$ so that $\hat \rho_1$ is a probability distribution, then the classical Schr\"odinger bridge problem over reweighted processes can be formulated as
	\begin{equation}\label{eq:SBPK}
		\min_{P\in\cP(\Omega)} \left\{H(P\mid \hat R )~\mid~ P_0 = \rho_0,\; P_1 = \hat\rho_1\right\},
	\end{equation}
where 
	\begin{equation}
		\hat R = \exp\left(-\int_{0}^1V(t,X_t)dt\right)R
	\end{equation}
is the (unnormalized) distribution induced by the survival trajectories of the diffusion process \eqref{eq:diffusion} with killing rate $V$. 
The solution to this problem reads
	\begin{equation}
		P^\star=f(X_0)g(X_1)\hat R=f(X_0)\exp\left(-\int_{0}^1V(t,X_t)dt\right)g(X_1)R
	\end{equation}
where the two multipliers $f, g$ are chosen such that $P^\star$ satisfies the constraints $P_0 = \rho_0, P_1 = \hat \rho_1$. 
These two multipliers can again be obtained by solving a Schr\"odinger system. More specifically, let $\varphi, \hat\varphi$ be the solution to
	\begin{subequations}
	\begin{eqnarray}
		\partial_t \hat\varphi &=& - \nabla\cdot(b \hat\varphi) - V\hat\varphi+\frac{1}{2} \sum_{i,j=1}^n \frac{\partial^2 (a_{ij} \hat \varphi)}{\partial x_i\partial x_j} 
		\\
		\partial_t\varphi &=& - b \cdot\nabla\varphi + V\varphi- \frac{1}{2} \sum_{i,j=1}^n a_{ij}\frac{\partial^2 \varphi}{\partial x_i\partial x_j}
		\\
		\rho_0 &=& \varphi(0,\cdot)\hat\varphi(0,\cdot)
		\\
		\hat \rho_1 &=& \varphi(1,\cdot)\hat\varphi(1,\cdot),
	\end{eqnarray}
	\end{subequations}
then $\varphi, \hat \varphi$ relate to $f,g$ as
	\begin{subequations}
	\begin{align}
		\hat\varphi(0,x)&=f(x)\hat R_0(x)\\
		\varphi(1,y)&=g(y).
		\end{align}
	\end{subequations}
	
Unlike the solution $\bP^\star$ to the uSBP \eqref{eq:USBP}, the solution $P^\star$ to \eqref{eq:SBPK} is a probability measure over $\Omega = C([0,1], \R^n)$. Indeed, it is associated with the diffusion process 
	\[
		dX_t = (b(t,X_t) +a(t,X_t)\nabla\log\varphi(t,X_t))dt + \sigma(t,X_t) dW_t
	\]
without losses. The marginal distribution of it equals $P_t = \varphi(t,\cdot) \hat\varphi(t,\cdot)$ and is a probability measure over $\R^n$ for all $t\in [0,\,1]$. We argue that the SBP over weighted process doesn't address Schr\"odinger's orginal problem as described in the thought experiment in Section \ref{sec:uSBP}. The prior $\hat R$ describes the distribution of the surviving trajectories and the problem \eqref{eq:SBPK} can be interpreted as finding the most likely evolution of surviving trajectories that are compatible with the two marginals $\rho_0, \hat\rho_1$. However, the mechanism of how the particles that did not survive got killed is completely ignored in this formulation. 

The importance of explicitly considering a possible update of the killing rate becomes salient when the end-point marginals are consistent with the prior law. Such a case highlights a dichotomy between our formulation of uSBP, and the rationale behind SBP over reweighted processes
To see this, consider a scenario where the two marginals are already consistent with the prior law, that is
	\[
		\rho_1(\cdot) = \int_{\R^n} q(0,x,1,\cdot) \rho_0(x) dx.
	\]
One would expect the solution to be the prior $\hat R$, {\em since the prior is consistent} with the end-point marginals. {\em This is, however, not the case!} Indeed, $\hat R_0$ represents the distribution at $t=0$ of those particles that are destined to survive, and this differs from $\rho_0$, the distribution of all particles. Thus, $\hat R_0$ is not the solution to \eqref{eq:SBPK}.

One could attempt to modify Schr\"odinger's thought experiment by postulating that $\rho_0$ is precisely the distribution at $t=0$ of those particles that eventually survive. With this modification, it is easy to see that the prior $\hat R$ solves \eqref{eq:SBPK}. {\em This modification, however, is not physical}: It is not possible to measure at time $t=0$ the marginal of the survival trajectories! 
	
Finally, we note that the Schr\"odinger bridge problem over reweighted processes has the following fluid dynamic (stochastic control) formulation
	\begin{subequations}\label{eq:fluidV}
	\begin{eqnarray}
		\min_{P_t(\cdot), u(t,\cdot)} && \int_0^1 \int_{\R^n} [\frac{1}{2} \|u(t,x)\|^2+V(t,x)] P_t dx dt
		\\&& \partial_t  P_t+ \nabla\cdot((b+\sigma u) P_t) - \frac{1}{2} \sum_{i,j=1}^n \frac{\partial^2 (a_{ij} P_t)}{\partial x_i\partial x_j}=0
		\\&& P_0 = \rho_0,\quad P_1=\hat\rho_1.
	\end{eqnarray}
	\end{subequations}
This stochastic control problem is over the diffusion process without losses
	\[
		dX_t = b(t,X_t)dt +\sigma(t,X_t)u(t,X_t)dt + \sigma(t,X_t) dW_t,
	\]
and the control $u(t,x)$ only enters the system through the drift. The prior killing rate $V$ serves as a cost term.
This is substantially different from the control formulation \eqref{eq:fluidkilling} of the uSBP where the control has a drift term $u(t,x)$ and a correcting term $\alpha(t,x)$, and the killing rate $V$ appears in the dynamics instead of the cost function.

\section{Numerical example}\label{sec:example}

We conclude by highlighting the uSBP formalism with an academic/numerical example. To this end, we consider the diffusion process
	\[
		dX_t= \sigma dW_t,
	\]
with $X_t,W_t\in \mathbb R$ (i.e., in a $1$-dimensional state space), $\sigma=0.05$, and 
killing rate
	\[
		V(t,x) \equiv 1.
	\]
	We work out the solution of the unbalanced Schr\"odinger bridge problem (uSBP) with initial marginal density
	\[
		\rho_0(x) =
        \begin{cases}
        {0.3-0.3\cos(3\pi x)} & \text{if}~ 0\le x<2/3\\
        {2.4-2.4\cos(6\pi x-4\pi)} & \text{if}~ 2/3\le x\le 1,
        \end{cases}
	\]
and target marginal density 
    \[
        \rho_1(x)= s \rho_0(1-x),
    \]
where $s\le 1$ denotes the percentage of survival particles. 

Figures \ref{fig:SBP} and  \ref{fig:muSBP} display the marginal flow of the uSBP for different values of $s$. 
When $s<1$, only a portion of the particles survive until the end and many particles vanish along the way. Thus, the total mass of the particles is a decreasing function of time, as can be seen from Figure \ref{fig:uSBP}. Note that the terminal percentage of surviving particles is consistent with the chosen value for $s$, in each case.
\begin{figure}\begin{center}
    \includegraphics[width=0.40\textwidth]{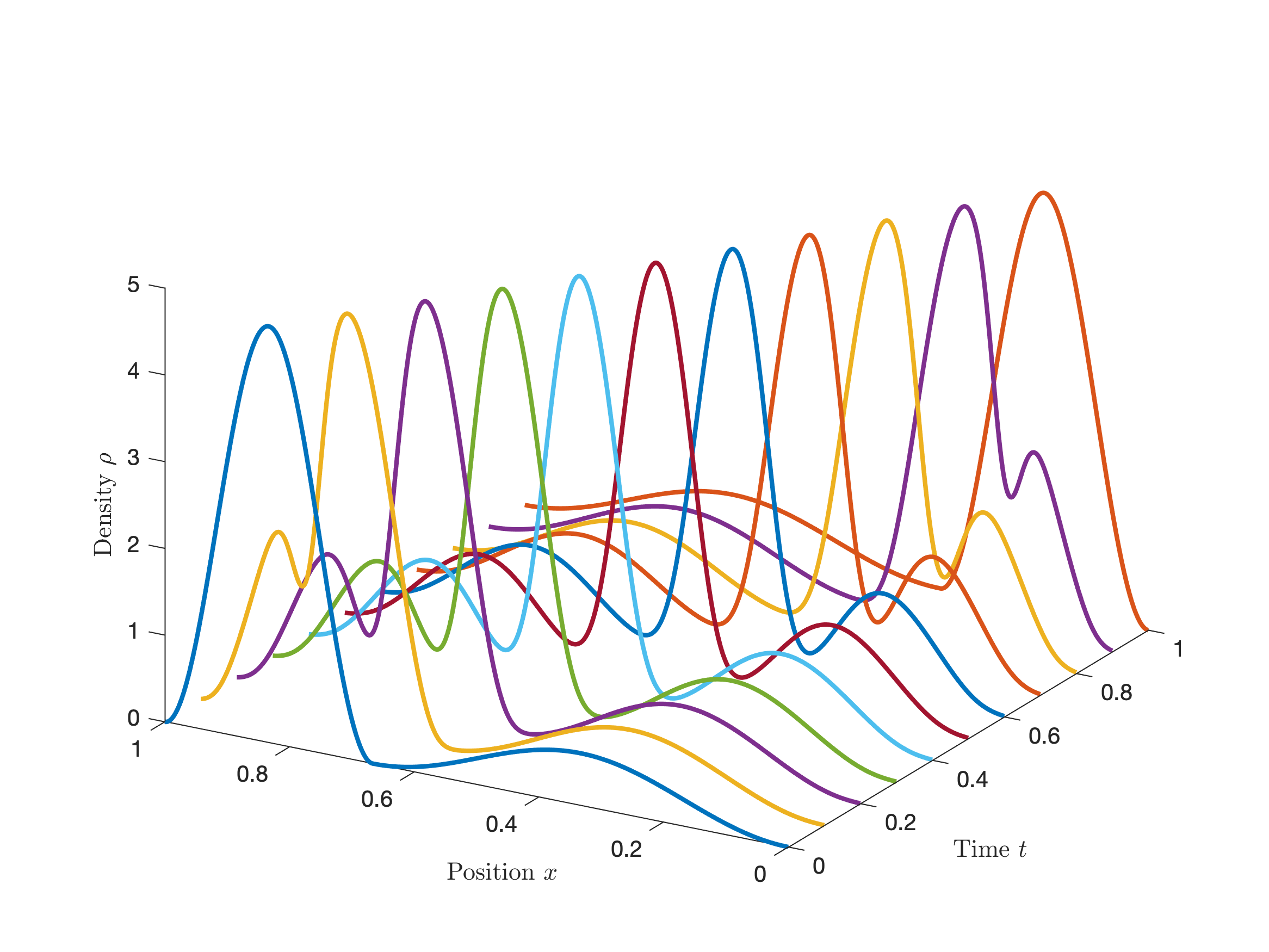}
    \caption{Marginal flow of uSBP for $s=1$}
    \label{fig:SBP}
\end{center}\end{figure}

\begin{figure}
     \centering
     \begin{subfigure}[b]{0.32\textwidth}
         \centering
         \includegraphics[width=\textwidth]{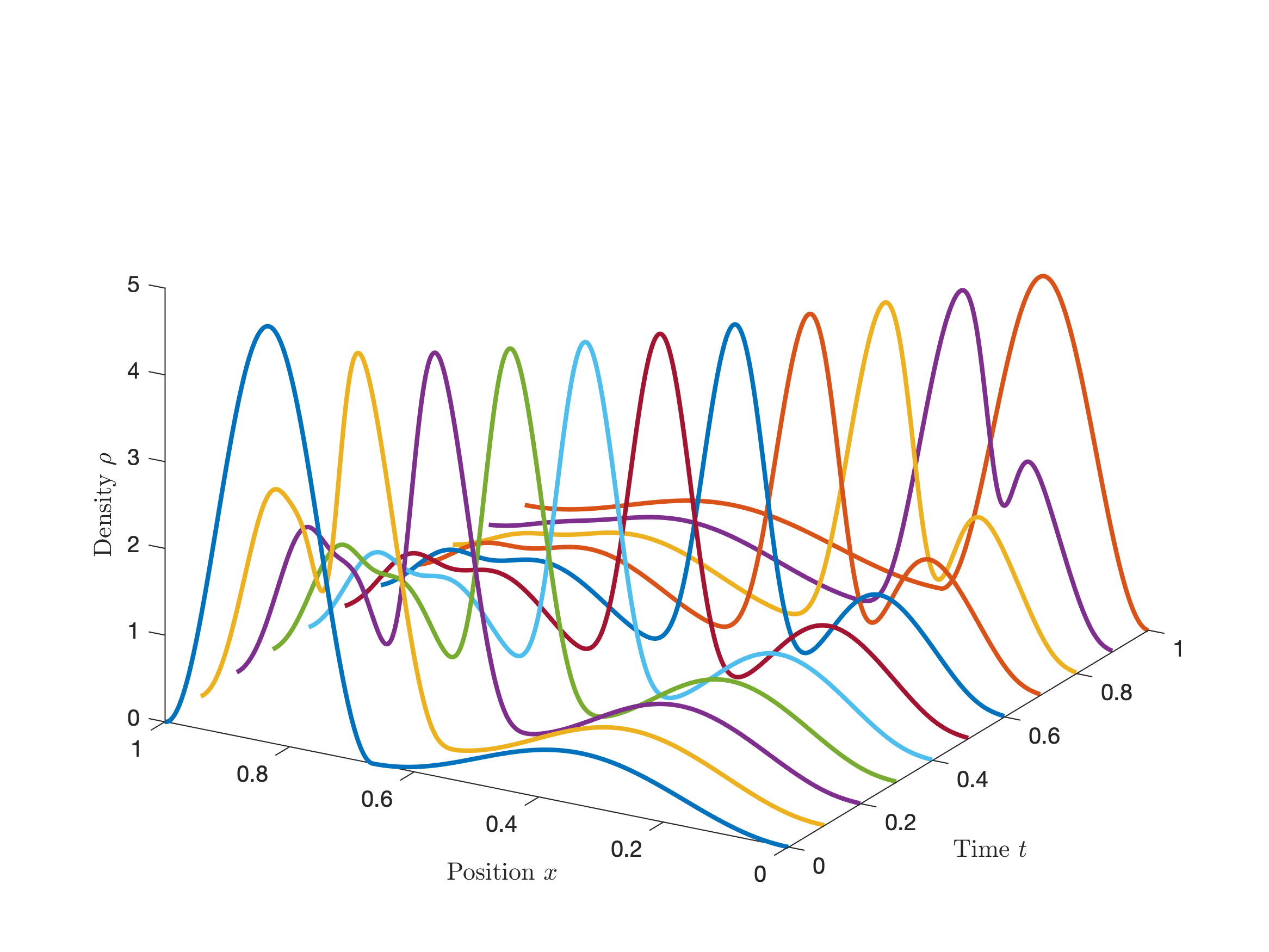}
         \caption{$s=0.8$}
         \label{fig:ms8}
     \end{subfigure}
     \hfill
     \begin{subfigure}[b]{0.32\textwidth}
         \centering
         \includegraphics[width=\textwidth]{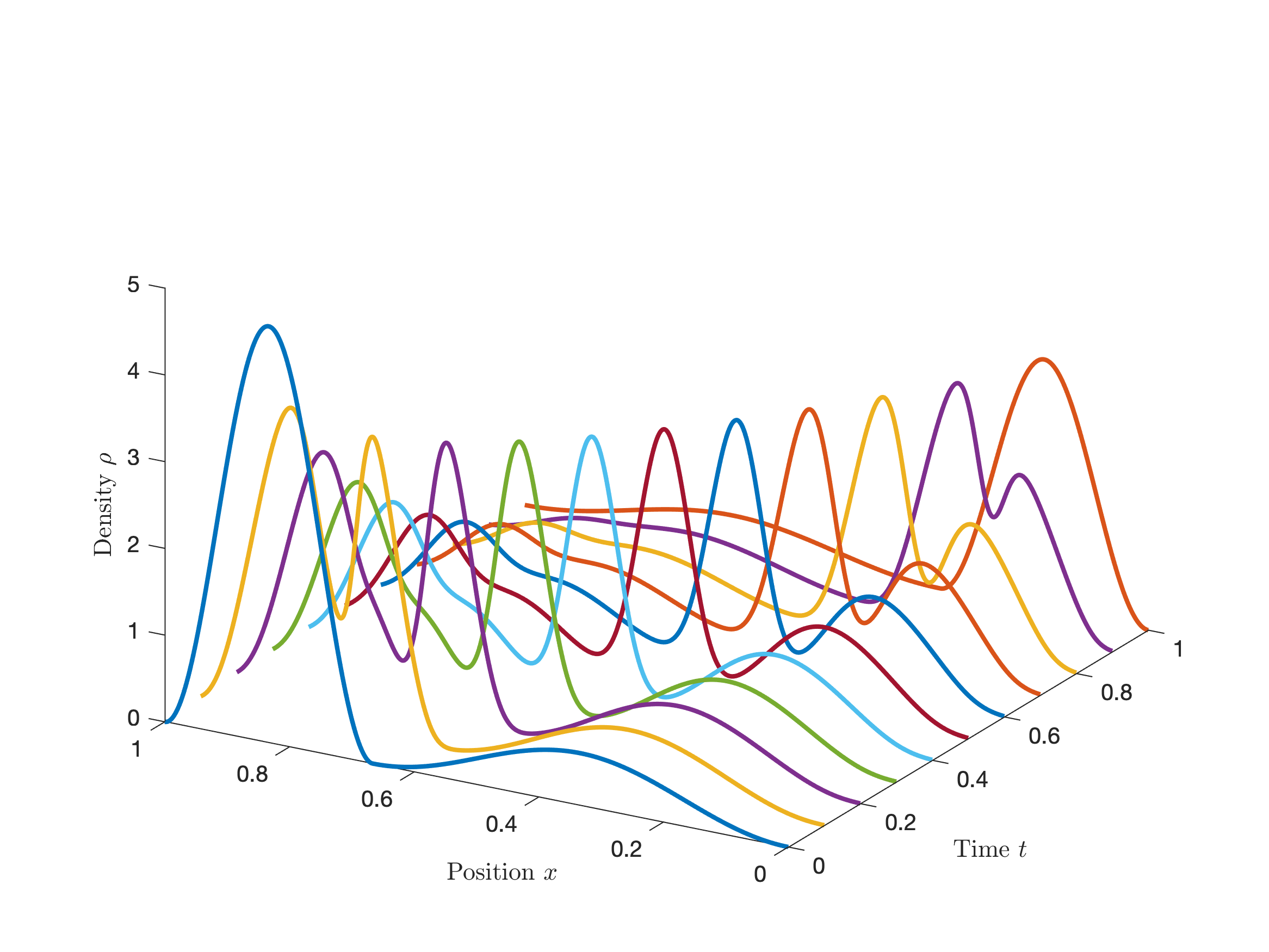}
         \caption{$s=0.6$}
         \label{fig:ms6}
     \end{subfigure}
     \hfill
     \begin{subfigure}[b]{0.32\textwidth}
         \centering
         \includegraphics[width=\textwidth]{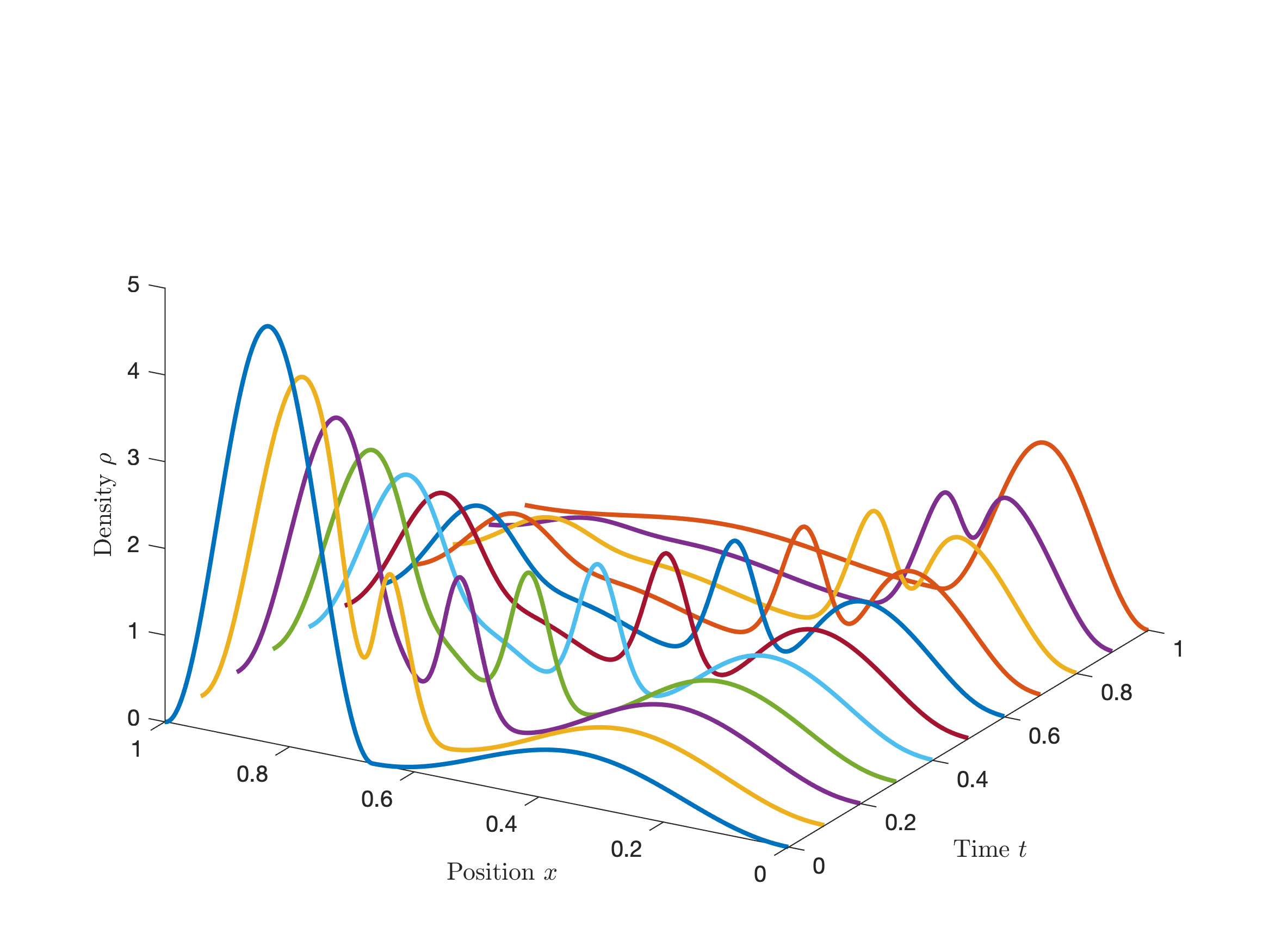}
         \caption{$s=0.4$}
         \label{fig:ms4}
     \end{subfigure}
        \caption{Marginal flow of uSBP for $s=0.8, 0.6, 0.4$}
        \label{fig:muSBP}
\end{figure}

\begin{figure}
     \centering
     \begin{subfigure}[b]{0.32\textwidth}
         \centering
         \includegraphics[width=\textwidth]{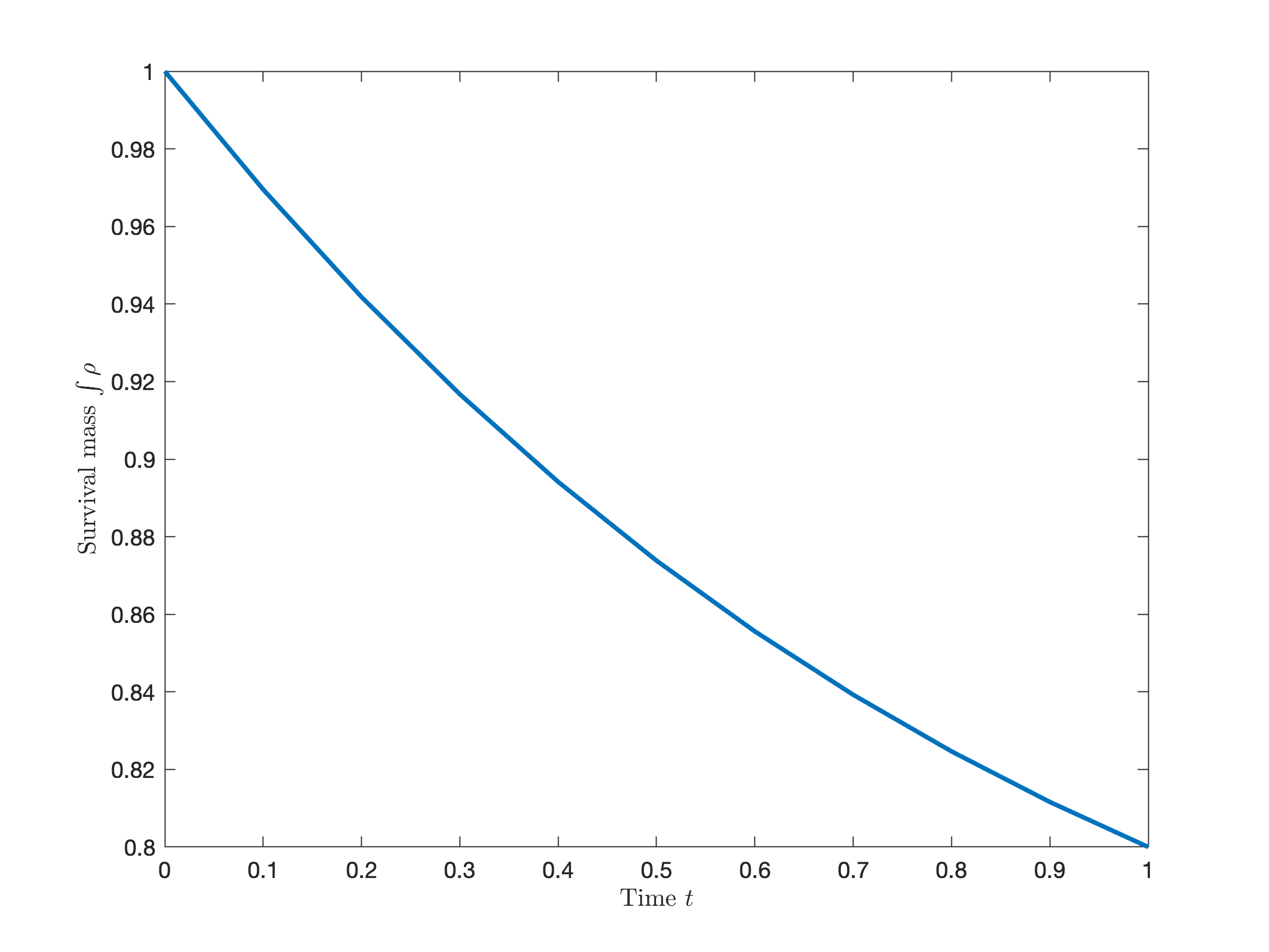}
         \caption{$s=0.8$}
         \label{fig:s8}
     \end{subfigure}
     \hfill
     \begin{subfigure}[b]{0.32\textwidth}
         \centering
         \includegraphics[width=\textwidth]{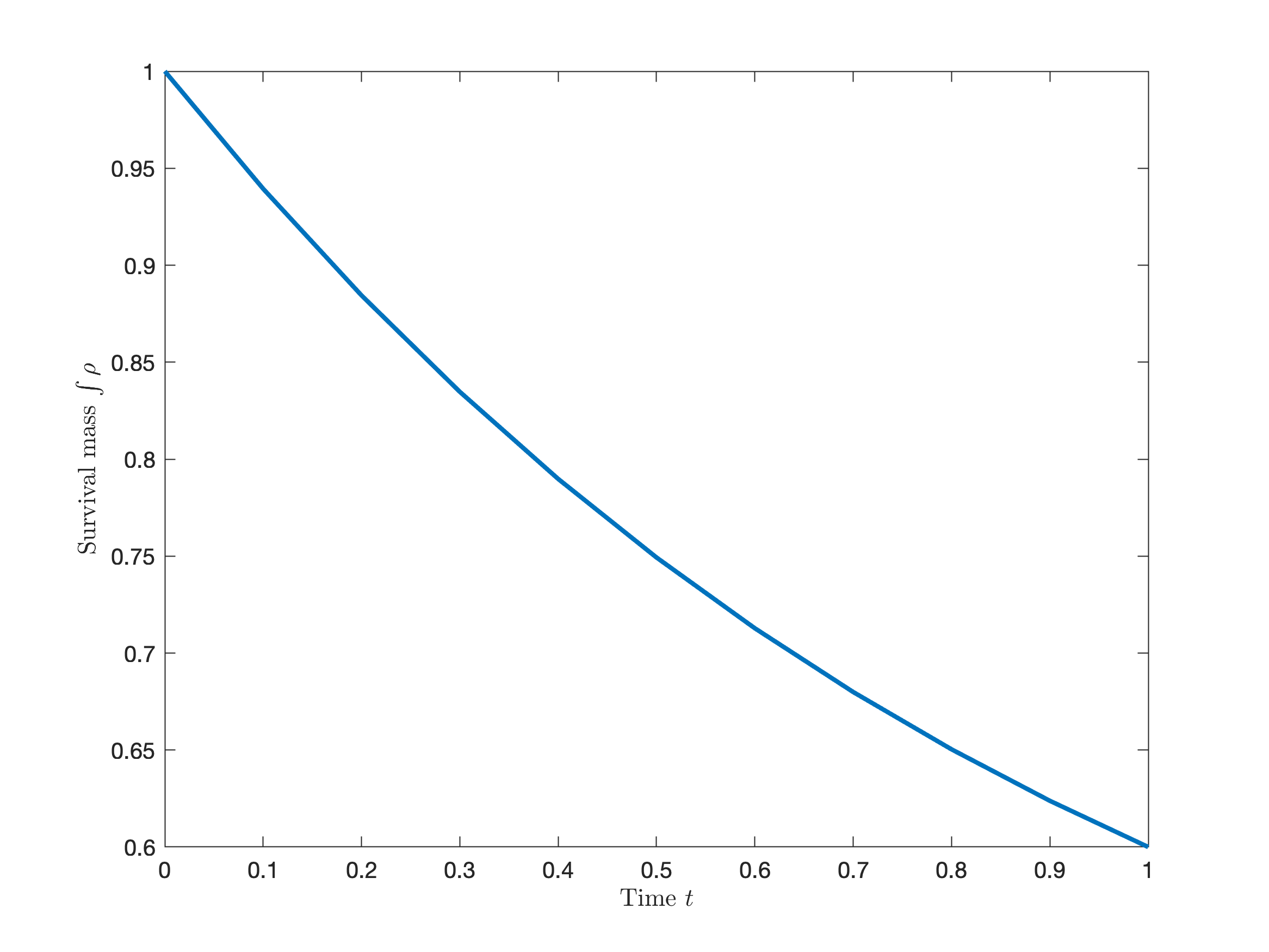}
         \caption{$s=0.6$}
         \label{fig:s6}
     \end{subfigure}
     \hfill
     \begin{subfigure}[b]{0.32\textwidth}
         \centering
         \includegraphics[width=\textwidth]{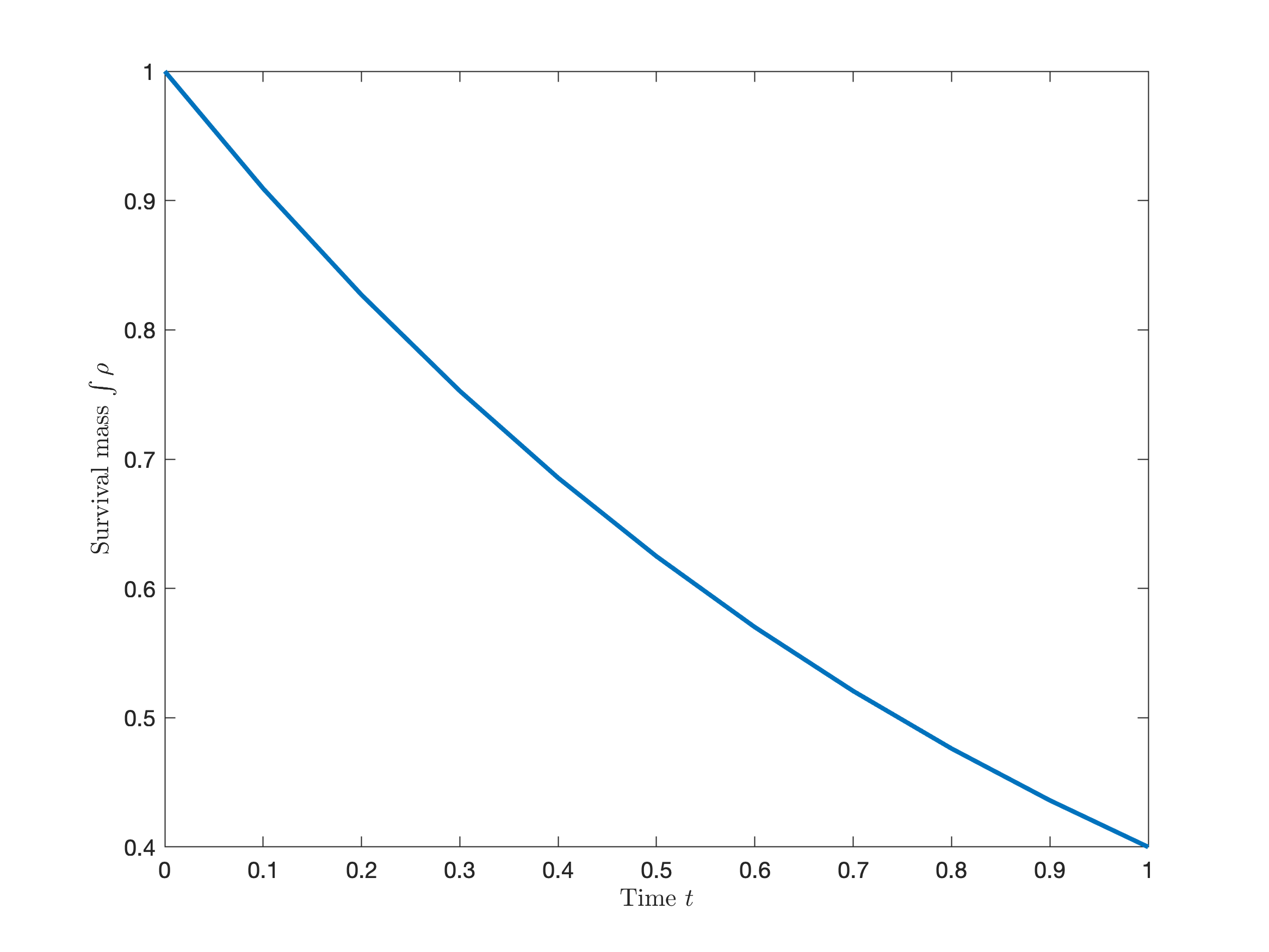}
         \caption{$s=0.4$}
         \label{fig:s4}
     \end{subfigure}
        \caption{Survival mass of uSBP for $s=0.8, 0.6, 0.4$}
        \label{fig:uSBP}
\end{figure}
For comparison, we also display the solution to the SBP over reweighted processes in Figure \ref{fig:reweighted}. Note that its solution is independent of $s$. The solution describes the posterior distribution of the survived particles only, and thus the marginal flow remains a probability measure at all times. In fact, it coincides with uSBP for $s=1$. 
\begin{figure}\begin{center}
    \includegraphics[width=0.40\textwidth]{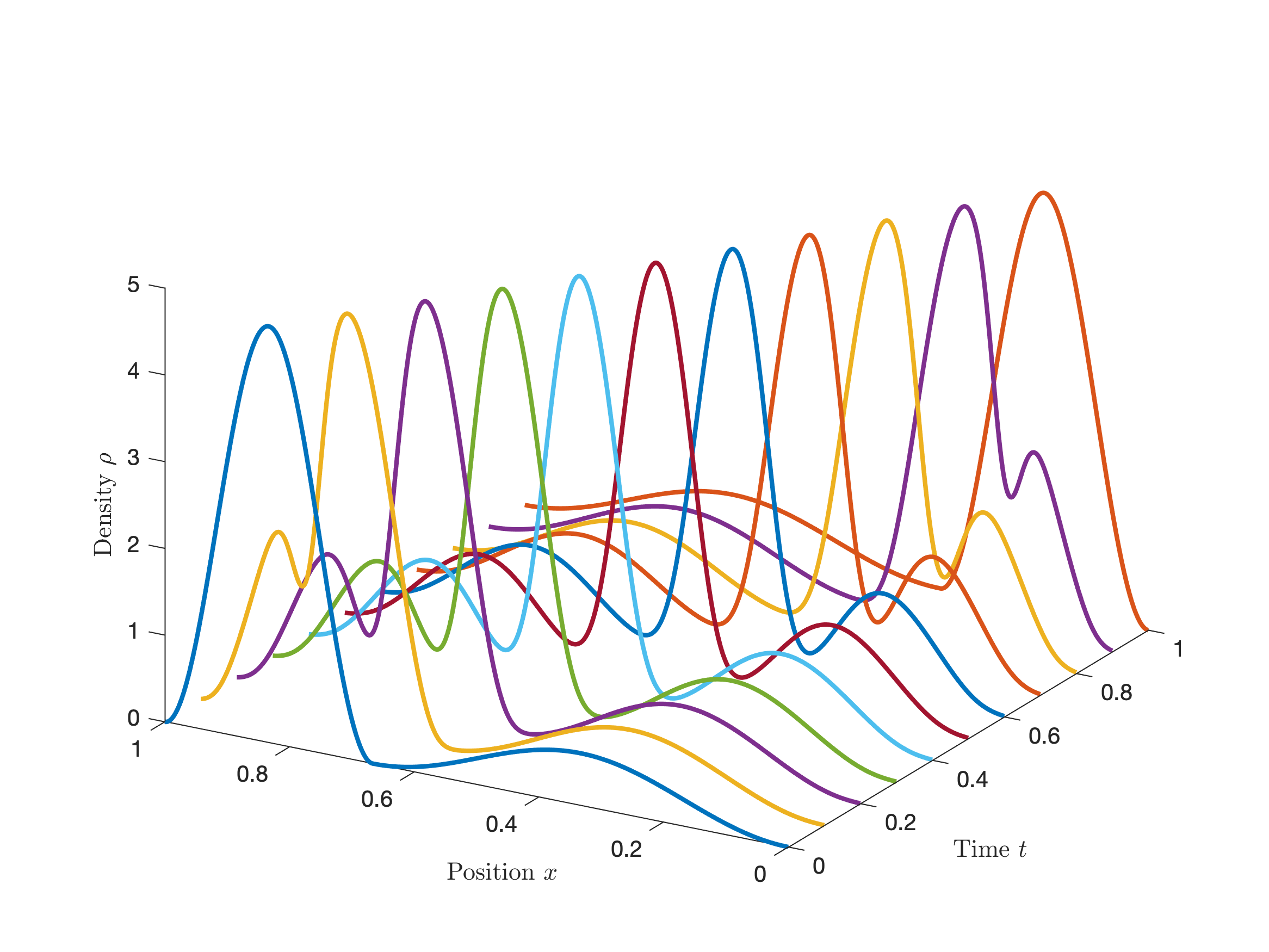}
    \caption{Marginal flow of SBP over reweighted processes}
    \label{fig:reweighted}
\end{center}\end{figure}

\section{Concluding remarks}\label{sec:conclusion}
We introduced Schr\"odinger bridges between unbalanced marginals in the spirit of E.\ Schr\"odinger's original rationale (that led to the standard SBP), aimed to reconcile a given prior law, that now includes a killing rate, with marginal observations.
We formulated the problem as maximum entropy problem over an augmented state space that includes a coffin state representing the state of vanishing particles. The solution is characterized by a Schr\"odinger-type system, different to the classical one, that yields a diffusion process whose drift as well as killing rate are suitable adjusted as compared to the prior. Just like in the standard SBP, this new unbalanced Schr\"odinger bridge problem (uSBP) can be formulated as a stochastic control problem. Naturally, departing from the standard SBP, the control variable in this control formulation includes both the drift and killing rate. We underscore an apparent dichotomy between our formulation of the uSBP and earlier work on SBP over reweighted processes with Feynman-Kac functionals. Though both pertain to SBP's for diffusions with losses, we argued that our uSBP is a natural formulation in the spirit of Schr\"odinger's original quest to reconcile probabilistic models with observations.
The nature of the zero-noise limit of the uSBP and its relation to a corresponding optimal transport flow between unbalanced marginals is left as a topic of future research.

\section{Appendix A: Hilbert's projective metric}\label{sec:appendix}

Herein we discuss Hilbert's projective metric and highlight some important contraction theorems due to Garrett Birkoff and P. J.\ Bushell \cite{birkhoff1957extensions,bushell1973projective,bushell1973hilbert} that we use in this work.  A first application of the Birkhoff-Bushell contractive maps  to scaling of nonnegative matrices, a topic closely connected to Schr\"odinger bridges, was presented in \cite{franklin1989scaling}. In \cite{georgiou2015positive}, it was  shown that the Schr\"odinger bridge for Markov chains and quantum channels can be efficiently obtained from the fixed-point of a map which  contracts the  Hilbert metric.  We refer to \cite{chen2016entropic,SIREV} for more detailed information and further applications of this metric.
Below, following \cite{bushell1973hilbert}, we recall some basic concepts and results of this theory.\\

Let $\mathcal B$ be a real Banach space and let $\cK$ be a closed solid cone in $\mathcal B $, i.e., $\cK$ is closed with nonempty interior $\cK_0$ and is such that $\cK+\cK\subseteq \cK$, $\cK\cap -\cK=\{0\}$ as well as $\lambda \cK\subseteq \cK$ for all $\lambda\geq 0$. Define the partial order
\[
\bx\preceq \by \Leftrightarrow \by-\bx\in\cK,\quad \bx \prec \by \Leftrightarrow \by-\bx\in\cK_0
\]
and for $\bx,\by\in\cK^+:=\cK\backslash \{0\}$, define
\begin{eqnarray*}
M(\bx,\by)&:=&\inf\, \{\lambda\,\mid \bx\preceq \lambda \by\}\\
m(\bx,\by)&:=&\sup \{\lambda \mid \lambda \by\preceq \bx \}.
\end{eqnarray*}
Then, the Hilbert metric is defined on $\cK^+$ by
\[
d_H(\bx,\by):=\log\left(\frac{M(\bx,\by)}{m(\bx,\by)}\right).
\]

It is easily seen that $d_H(\cdot,\cdot)$ is symmetric, i.e., that $d_H(\bx,\by)=d_H(\by,\bx)$, and invariant under scaling by positive constants, since $d_H(\bx,\by)=d_H(\lambda \bx, \by)$ for any $\lambda>0$ and $\bx,\by\in\cK_0$.
Therefore $d_H(\lambda \bx,\bx)=0$. It can also be shown that the triangular inequality holds and, therefore, $d_H(\cdot,\cdot)$  is a {\em projective} metric that represents distance
between rays.

In our analysis we encounter two types of maps. We encounter inversion 
\begin{align}\label{eq:isometry}
\mathcal E_{\rm inv}\,:\; \bx\mapsto \bx^{-1},
\end{align}
of elements $\bx\in\cK_0$, and also linear maps that are {\em positive}, namely,
\[
\cE:\cK^+\rightarrow\cK^+.
\]
For both types of maps we are interested in
determining their {\em contraction ratio}, 
\begin{eqnarray*}
\kappa(\cE):=\inf\{\lambda \mid d_H(\cE(\bx),\cE(\by))\leq \lambda d_H(\bx,\by),\forall \bx,\by\in %{\rm int}\,
\cK_0\}.
\end{eqnarray*}
It turns out that the former are isometries in the Hilbert metric whereas the latter are contractions. Thus, the composition of a combination of both types turns out to be a contraction.

That \eqref{eq:isometry} is an isometry, i.e., $\kappa(\mathcal E_{\rm inv})=1$, follows immediately from
\[
M(\bx,\by)=\frac{1}{m(\bx^{-1},\by^{-1})} \label{eq:isometry2}.
\]
Then, by G.\ Birkhoff's theorem  \cite{birkhoff1957extensions,bushell1973hilbert}, any positive linear map $\mathcal E$
is contractive and the contraction ratio can be expressed in terms of the {\em projective diameter} 
\begin{eqnarray*}
\Delta(\cE):=\sup\{d_H(\cE(\bx),\cE(\by))\mid \bx,\by\in \cK_0\}
\end{eqnarray*}
of the range of $\mathcal E$.
Specifically, under these conditions, G.\ Birkhoff's theorem states that
\begin{equation}\label{condiam}
\kappa(\cE)=\tanh(\frac{1}{4}\Delta(\cE)).
\end{equation}
Thus, a positive linear map is strictly contractive if its projective diameter $\Delta(\cE)$ is finite.
A further useful observation, that follows from the triangular inequality, is that for any $\bx_0\in\mathcal \cK_0$,
\begin{align}\label{eq:useful}
\Delta(\cE)&\leq 2\sup\{d_H(\mathcal E(\bx),\bx_0) \mid \bx\in\mathcal K_0\}.
\end{align}
This allows bounding $\Delta(\cE)$ to ensure strict contraction for $\mathcal E$.

Important examples are provided by the positive orthant of $\R^n$, the cone of Hermitian, positive semidefinite matrices, spaces of bounded positive functions, and so on.
Notice that, in all cases, the boundary of the cone lies at an infinite distance from any interior point.

\section{Appendix B: Proof of Theorem \ref{thm:SBEsys} on the generalized Schr\"odinger system}\label{sec:GSS}
We herein establish existence and uniqueness of solution (up to constant positive scaling) for the system (\ref{eq:SBsystem}).
The steps mimick the analogous case for the SBP where the marginals are supported on a Euclidean space \cite{chen2016entropic}. The difference at present lies in that the support of functions includes an added point that represents the coffin state.

We assume throughout that the marginal measures $\rho_0,\rho_1$ are absolutely continuous with respect to the Lebesgue measure, in that $\rho_0(dx)=\rho_0(x)dx$ and $\rho_1(dx)=\rho_1(x)dx$ for density functions $\rho_0,\rho_1$ with support $S_0,S_1\subseteq \mathbb R^n$, respectively, and that $\rho_0$ is a probability measure %on (the Borel sets of) $S_0$
while $\rho_1$ is a nonnegative measure %on (the Borel sets of) $S_1$ such that
with $\int_{S_1}\rho_1(x)dx\le 1$. 
The case $\int_{S_1}\rho_1(x)dx= 1$ reduces to the standard SBP and is easy to handle. Thus, without loss of generality, we assume
	\[
		\int_{S_1}\rho_1(x)dx> 1.
	\]

As before, we let $q(0,x_0,t,x)$ for $0<t\le 1$ denote the fundamental solution of equation (\ref{eq:SBsystema}) and assume that
%$q(0,x_0,1,x_1)$ 
it is continuous and strictly positive on compact subsets. This is guaranteed by sufficient smoothness of the coefficients $b,V,a$, positivity of $V$ and positive definiteness on the whole domain of the matrix $a=(a_{ij})$. 
Under these assuptions we rewrite the Schr\"odinger system (\ref{eq:SBsystem}) as follows,
\begin{subequations}\label{eq:SBsystem2}
	\begin{align}
	\label{eq:SBsystem2a}
		 \hat{\varphi}(t,x)&= \int_{\mathbb R^n} q(0,x_0,t,x) \hat{\varphi}(0,x_0)dx_0,
		\\\label{eq:SBsystem2b}
		\hat\psi(1) &= \int_0^1\int_{\mathbb R^n} V(t,x)\hat\varphi(t,x) dxdt
		\\\label{eq:SBsystem2c}
		\varphi(0,x_0)&=\int_{\mathbb R^n} q(0,x_0,1,x_1)\varphi(1,x_1)dx_1+\int_0^1\int_{\mathbb R^n}q(0,x_0,t,x)V(t,x)\psi(t) dxdt
		\end{align}
\begin{align}
\label{eq:SBsystem2d}\nonumber
&\phantom{	=\int_{S_1} q(0,x_0,1,x_1)\varphi(1,x_1)dx_1+\int_0^1\int_{S_1}q(0,x_0,t,x)V(t,x)\psi(t) dxdt}\\[-40pt]
		\psi(t) &= {\rm constant},\quad 0\le t\le 1,
		\\\label{eq:SBsystem2e}
		\rho_0(x_0) &= \varphi(0,x_0)\hat\varphi(0,x_0)
		\\\label{eq:SBsystem2f}
		\rho_1(x_1) &= \varphi(1,x_1)\hat\varphi(1,x_1)
		\\\label{eq:SBsystem2g}
		  \psi(0)\hat\psi(0) &= 1-\int_{S_0}\rho_0=0
		 \\\label{eq:SBsystem2h}
		 \psi(1)\hat\psi(1) &= 1-\int_{S_1} \rho_1 =: c_1 > 0.
	\end{align}
	\end{subequations}
We consolidate the system of equations \eqref{eq:SBsystem2} into
\begin{subequations}\label{eq:SBsystem3}
	\begin{align}\label{eq:SBsystem3a}
		 \hat{\varphi}(t,x)&= \int_{S_0} q(0,x_0,t,x) \rho_0(x_0) \frac{1}{\varphi(0,x_0)}dx_0
		\\\nonumber
		\hat\psi(1) &= \int_0^1\int_{\mathbb R^n} V(t,x)\hat\varphi(t,x) dxdt\\\label{eq:SBsystem3b}
		&= \int_{S_0}\frac{1}{\varphi(0,x_0)}\rho_0(x_0)\underbrace{\int_0^1\int_{\mathbb R^n}  q(0,x_0,t,x) V(t,x) dxdt}_{r(x_0)}dx_0
		\\\label{eq:SBsystem3c}
		\varphi(0,x_0)&=\int_{S_1} q(0,x_0,1,x_1)\rho_1(x_1)\frac{1}{\hat\varphi(1,x_1)}dx_1+ \frac{1}{\hat\psi(1)}c_1\underbrace{\int_0^1\int_{\mathbb R^n} q(0,x_0,t,x)V(t,x) dxdt}_{r(x_0)}\\\label{eq:SBsystem3d}
		\psi(0) &=\psi(1)= \frac{1}{\hat\psi(1)}c_1.
			\end{align}
	\end{subequations}

These four equations, that encapsulate the Schr\"odinger system, suggest considering the composition of maps
    \begin{align}\label{eq:newmaps}
	{\hat\varphi(1,\cdot)\choose \hat\psi(1)}
        			\stackrel{\mathcal E_1}{\mapsto}
       {\frac{1}{\hat\varphi(1,\cdot)}\choose \frac{1}{\hat\psi(1)}}
       			\stackrel{\mathcal E_2}{\mapsto} 
			       {\varphi(0,\cdot) \choose \psi(0)}
       			\stackrel{\mathcal E_3}{\mapsto}
        {\frac{1}{\varphi(0,\cdot)}\choose \frac{1}{\psi(0)}}
       			 \stackrel{\mathcal E_4}{\mapsto} 
      {\hat\varphi(1,\cdot)\choose \hat\psi(1)}_{\rm next}
    \end{align}
in order to analyze existence of solutions. Indeed, we utilize the theory of the Hilbert metric (outlined in Appendix \ref{sec:appendix}) to show that the composition is a strict contraction along rays, resulting in a unique fixed point.
 
To this end, we consider the Banach space $\mathcal B=\mathcal L^\infty(\mathcal X)$ of {\em real-valued} functions $h(\cdot)$ on
$\cX =S\cup \{\mathfrak c\}$, where $S\in\R^n$ satisfies that $S_0\cup S_1\subset S$.
For notational convenience we use the vectorial notation ${h(x)\choose h(\mathfrak c)}$ to specify the values of $h$ on the two constituents of its support, for $x\in\R^n$ and $\mathfrak c\in\{\mathfrak c\}$. 
Thus, the norm of $h$ is
\[
\left\|h\right\|:=\max\{\|h|_S\|_\infty,|h(\mathfrak c)|\}.
\]
We consider the cone of positive functions
\[
\cK=\{ h \in \mathcal B\mid h(\mathfrak c)\geq 0 \mbox{ and } h(x)\geq 0 \; a.e.\ x\in S\}
\]
and the corresponding partial order $h_1\preceq h_2 \Leftrightarrow h_2-h_1\in\cK$ as usual.
We observe that $\cK$ is closed, solid and has a non-empty interior (of strictly positive a.e.\ functions) that we denote $\cK_0$; we also denote $\cK^+:=\cK\backslash \{0\}$. 

Note that in the on-going development, the components of functions $h\in\mathcal B$, that are (possibly time-dependent) functions on $\R^n$ and $\{\mathfrak c\}$, respectively, are differentiated as $\varphi,\psi$, or  $\hat\varphi,\hat\psi$, respectively, e.g., ${h(x)\choose h(\mathfrak c)}={\varphi(t,x)\choose \psi(t)}$.
We proceed to consider the composition of maps in \eqref{eq:newmaps} and establish first the following weaker version of Theorem \ref{thm:SBEsys}:\\

\begin{theorem}\label{thm:SBEsys_weak} Assuming that the support sets $S_0,S_1$ of the two marginals $\rho_0,\rho_1$ of the uSBP are compact, the claim in Theorem \ref{thm:SBEsys} holds true.
\end{theorem}

Recall the notation $M(\cdot,\cdot),m(\cdot,\cdot),\kappa(\cdot)$ and $\Delta(\cdot)$ from Appendix \ref{sec:appendix}.
As noted in the appendix, since $M(h_1,h_2)=m(h_1^{-1},h_2^{-1})^{-1}$ for $h_1,h_2\in\cK_0$,
both $\mathcal E_1$ and $\mathcal E_3$ are isometries. They are readily extended to isometries on $\mathcal K^+$ as well.

The map $\mathcal E_2$ is linear (homogeneous of degree $1$) and therefore, by Birkhoff's theorem given in the appendix, contractive on $\mathcal K^+$.
For the same reason, $\mathcal E_4$ is contractive. 
Unfortunately, neither map is strictly contractive. To see this, note that since, e.g., $\mathcal E_2 ({\star \choose 0})={\star \choose 0}$, with $\star$ denoting nonzero entries, certain elements on the boundary of $\mathcal K^+$ map onto the boundary and not the interior.

In order to establish the theorem we proceed as follows. Let $z\in S_0$ be an arbitrary fixed point in $S_0$. We modify equation \eqref{eq:SBsystem3d} of the Schr\"odinger system \eqref{eq:SBsystem3}, replacing it with
\begin{equation}\tag{\ref{eq:SBsystem3d}'} \label{eq:modification}
\tilde\psi(0)= \varphi(0,z) = \int_{S_1} q(0,z,1,x_1)\rho_1(x_1)\frac{1}{\hat\varphi(1,x_1)}dx_1+ \frac{1}{\hat\psi(1)}c_1 r(z),
\end{equation}
and, accordingly, replace $\mathcal E_2$ with a corresponding map that we refer to as $\mathcal E_2^\prime$. We then show the existence and uniqueness of solution for the modified system. Interestingly, except for the last of the elements in the $4$-tuple $(\hat\varphi(1,x), \hat \psi(1), \varphi(0,x), \psi(0))$, namely, $\psi(0)$, the remaining dictate the sought solution of the original Schr\"odinger system \eqref{eq:SBsystem}. This last entry plays no role in the original Schr\"odinger system. In particular,
	\begin{equation}\label{eq:equivalence}
		\mathcal E_4\circ\mathcal E_3\circ\mathcal E_2\circ\mathcal E_1 = \mathcal E_4\circ\mathcal E_3\circ\mathcal E_2^\prime\circ\mathcal E_1=:\cC.
	\end{equation}

We now consider $\mathcal E_2^\prime : h\mapsto g$ and show that it is strictly contractive in the Hilbert metric.
From \eqref{eq:useful}, taking as $\bx_0$ the function which is identically equal to $1$ on $S$ as well as on
$\{\mathfrak c\}$, we deduce that
\begin{align}\label{eq:deltabound}
\Delta(\mathcal E_2^\prime)\leq 2\sup \{\log\left(\frac{\max\{\sup_xg(x),g(\mathfrak c)\}}{\min\{\inf_xg(x),g(\mathfrak c)\}}\right)\mid g=\mathcal E_2^\prime(h) \mbox{ and }h\in \cK_0\}
\end{align}
Since $\rho_0,\rho_1$ are supported on compact sets $S_0,S_1$ of $\mathbb R^n$, respectively, we can choose $S$ to be compact as well. Since the kernel $q$ is positive and continuous, the kernel is bounded from below and above on $S\times S$. I.e., there exist $0<\alpha_1\le \beta_1<\infty$ such that
    \begin{equation}\label{eq:ab}
        \alpha_1\le q(0,x,1,y) \le \beta_1,
    \end{equation}
   for all $(x,y)\in S\times S$. Similarly, there exist $0<\alpha_2\le \beta_2<\infty$ such that
    \begin{equation}\label{eq:ab2}
    	\alpha_2\le r(x) \le \beta_2
    \end{equation}
   for all $x\in S$. 
   
Let $h(x)=\frac{1}{\hat\varphi(1,x)}$ and $h(\mathfrak c)=\frac{1}{\hat\psi(1)}$, then
\[
\alpha_1 \int_{S_1} \rho_1(x_1)h(x_1)dx_1\leq   \int_{S_1} q(0,x_0,1,x_1)\rho_1(x_1)h(x_1)dx_1   \leq
\beta_1 \int_{S_1} \rho_1(x_1)h(x_1)dx_1, ~\forall x_0 \in S.
\]
%and similarly, since $r(\cdot)$ is likewise bounded above and below,
It follows that, in view of \eqref{eq:SBsystem3c},
\[
\frac{\sup_xg(x)}{\inf_xg(x)}
 \leq \frac{\max_{i\in\{1,2\}}\beta_i}{\min_{i\in\{1,2\}}\alpha_i} <\infty.
\]
Thanks to the modification \eqref{eq:modification}, $g(\mathfrak c) = g(z)$ and therefore
\[
\frac{\max\{\sup_xg(x),g(\mathfrak c)\}}{\min\{\inf_xg(x),g(\mathfrak c)\}}
 \leq \frac{\max_{i\in\{1,2\}}\beta_i}{\min_{i\in\{1,2\}}\alpha_i} <\infty.
\]
Thus, from \eqref{eq:deltabound} and using Birkhoff's theorem \eqref{condiam},
\[
\kappa(\mathcal E_2^\prime)<1.
\]
As a consequence, the composition $\mathcal E_4\circ \mathcal E_3\circ \mathcal E_2^\prime\circ \mathcal E_1$ is strictly contractive, i.e.,
\[
\kappa(\mathcal E_4\circ\mathcal E_3\circ\mathcal E_2^\prime\circ\mathcal E_1)<1.
\]
It follows from \eqref{eq:equivalence} that
\[
\kappa(\cC)= \kappa(\mathcal E_4\circ\mathcal E_3\circ\mathcal E_2\circ\mathcal E_1)<1.
\]	

The above condition ensures that $\cC$ has a unique fixed point in terms of the Hilbert metric \cite{chen2016entropic}. Since Hilbert metric is a projective metric, the uniqueness is up to a constant scaling. Denote the fixed point on the unit sphere $U$ by $h$, then 
	\[
		\cC(h) = \lambda h
	\]
for some positive number $\lambda$. We next show $\lambda=1$. To this end, we introduce a different factorization of $\cC$ as
	\[
		\cC = \cE^\dagger \circ \cE_{p_0} \circ \cE \circ \cE_{p_1},
	\]
where
	\begin{eqnarray*}
		\cE(u) & = & \left[\begin{matrix} 
		\int_{S_1} q(0,x,1,x_1)u(x_1)dx_1+ r(x)u(\mathfrak c)
		\\
		u(\mathfrak c)
		\end{matrix}\right]
		\\
		\cE_{p_0} (u) &=& \left[\begin{matrix} \frac{\rho_0(x)}{u(x)}\\0\end{matrix}\right]
		\\
		\cE_{p_1}(u) &=& \left[\begin{matrix} \frac{\rho_1(x)}{u(x)}\\\frac{c_1}{u(\mathfrak c)}\end{matrix}\right],
	\end{eqnarray*}
and $\cE^\dagger$ is the adjoint operator of $\cE$. 
Clearly,
	\[
		\langle u, \cE_{p_0}(u) \rangle = \langle \cE_{p_1} (u), u\rangle = 1, ~\forall u\in \cK_0.
	\]
It follows that
	\begin{eqnarray*}
		1 &=& \langle \cE \circ \cE_{p_1}(h), \cE_{p_0} \circ \cE\circ\cE_{p_1} (h)\rangle
		\\
		&=& \langle \cE_{p_1} (h), \cE^\dagger \circ \cE_{p_0} \circ \cE\circ\cE_{p_1}(h)\rangle
		\\
		&=& \langle \cE_{p_1} (h), \cC(h)\rangle
		\\
		&=& \langle \cE_{p_1} (h), \lambda h\rangle = \lambda.
	\end{eqnarray*}
Once the fixed point $h$ is computed, the $4$-tuple $(\hat\varphi(1,x), \hat \psi(1), \varphi(0,x), \psi(0))$ can be recovered by
	\[
		\hat\varphi(1,x) = h(x),\; \hat\psi(1) = h(\mathfrak c),
	\]
and
	\[
		\left[\begin{matrix} \varphi(0,\cdot)\\\psi(0)\end{matrix}\right]
		=
		\cE_2\circ\cE_1(h).
	\]
The uniqueness of the $4$-tuple $(\hat\varphi(1,x), \hat \psi(1), \varphi(0,x), \psi(0))$ follows from the uniqueness of the fixed point $h$. This completes the proof of Theorem \ref{thm:SBEsys_weak}.
A standard argument \cite[Theorem 3.5]{chen2016entropic} can be used to extend the proof to the setting where $S_0, S_1$ are not necessarily compact for Theorem \ref{thm:SBEsys}.

{
\bibliographystyle{IEEEtran}
\bibliography{./refs}
}
\end{document}